\documentclass[11pt]{article}
\usepackage{amssymb}
\usepackage{amsthm}
\usepackage{amsmath}
\newtheorem{thm}{Theorem}
\newtheorem{cor}[thm]{Corollary}
\newtheorem{lem}[thm]{Lemma}

\usepackage{tikz,tikz-3dplot}
    \usetikzlibrary{shapes.geometric}
\tikzstyle{scorestars}=[star, star points=5, star point ratio=2.25, draw,inner sep=1pt,anchor=center]
\usetikzlibrary{calc}
\usepackage{bbm}
\DeclareMathOperator*{\argmin}{arg\,min}
\DeclareMathOperator*{\argmax}{arg\,max}
\usepackage{mathtools}
\definecolor{COLOR1}{HTML}{e6194b}
\definecolor{COLOR2}{HTML}{0082c8}
\definecolor{COLOR3}{HTML}{f58231}
\definecolor{COLOR4}{HTML}{911eb4}
\usetikzlibrary{patterns}
\usepackage{pgfplots}
\usepackage[justification=centering]{subcaption}
\usepackage[margin = 1in]{geometry}
\usepackage{amsmath}
\usepackage{algorithm}
\usepackage[noend]{algpseudocode}

\makeatletter
\def\BState{\State\hskip-\ALG@thistlm}
\makeatother

\begin{document}
	
	%\address{$^{\ddag}$ Singapore University of Technology and Design.}
	%\EMAIL{georgios@sutd.edu.sg}
	\title{\vspace{-.7in}Fast and Furious Learning in Zero-Sum Games:\\
		%$\Theta(\sqrt{T})$
		Vanishing  Regret with Non-Vanishing Step Sizes}  
	\author{%
		{James P. Bailey} \phantom{and} {Georgios Piliouras}\\
		\normalsize Singapore University of Technology and Design\\
		\normalsize $\{$james\_bailey,georgios$\}$@sutd.edu.sg
	}
	\date{}
	\maketitle

	\begin{abstract}
		%Follow-the-regularized-leader algorithms, e.g., gradient descent, are basic tools of online learning and optimization theory. It is well known that FTRL guarantees a time-average regret that scales linearly with the step size (i.e. learning rate) and this guarantee is tight. 
		
		We show for the first time, to our knowledge, that it is possible to  reconcile in online learning in zero-sum games two seemingly contradictory objectives: vanishing time-average regret and non-vanishing step sizes. This phenomenon, that we coin   ``fast and furious" learning in games, sets a new benchmark about what is possible both in  max-min optimization as well as in multi-agent systems. 
		Our analysis does not depend on introducing a carefully tailored  dynamic. Instead we focus on the most well studied online dynamic, gradient descent.
		Similarly, we focus on the simplest textbook class of games,
		 		  two-agent two-strategy zero-sum games, such as Matching Pennies. 
		  Even for this simplest of benchmarks the best known bound for total regret, prior to our work, was the trivial one of $O(T)$, which is immediately applicable even to a non-learning agent.  
		  Based on a tight understanding of the geometry of the non-equilibrating trajectories in the dual space we prove
		a regret bound 
		 of $\Theta(\sqrt{T})$ matching the well known optimal bound for adaptive step sizes in the online setting. This guarantee holds for all fixed step-sizes without having to know the time horizon in advance and adapt the fixed step-size accordingly.
		 As a corollary, we establish that even with fixed learning rates the time-average of mixed strategies, utilities converge to their exact Nash equilibrium values.
		 
		 \bigskip
		 
	\end{abstract}
		\begin{figure}[h]
			\makebox[\textwidth][c]{
			\begin{tabular}{c c c}
				\includegraphics[scale=1.1]{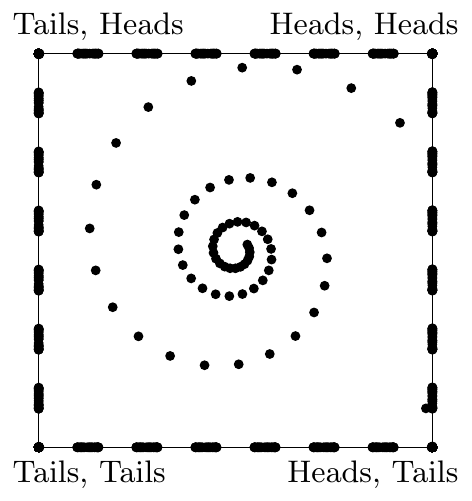}
				&
				\includegraphics[scale=1.1]{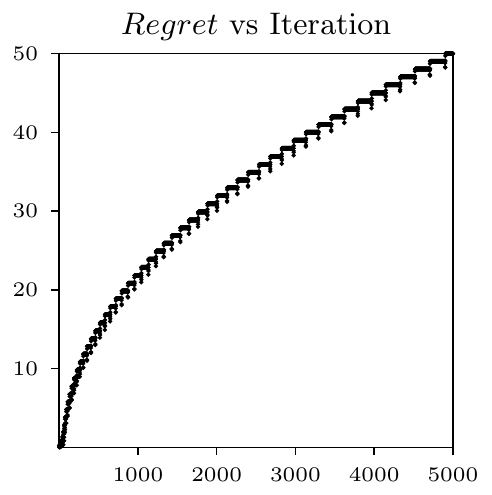}
				&
				\includegraphics[scale=1.1]{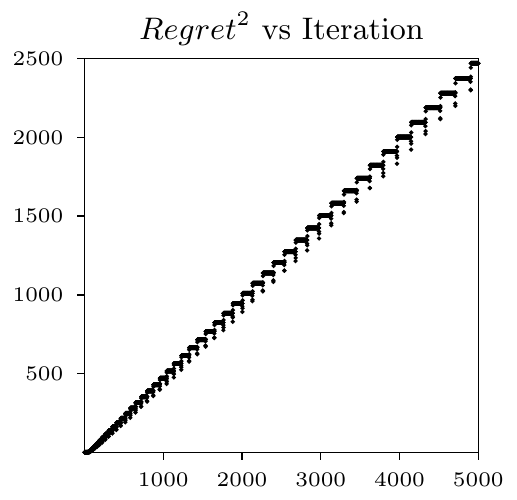}\\
				\small (a) Player Strategies
				&
				\small (b) Player 1 Regret
				& 
				\small (c) Player 1 Regret Squared
			\end{tabular}
			}
			\caption{5000 Iterations of Gradient Descent on Matching Pennies with $\eta=.15$.}\label{fig:Regret}
		\end{figure}

\newpage

%%%%%%%%%%%%%%%%%%%%%%%%%%%%%%%%%%%%%%%%%%%%%%%%%%%%%%
\section{Introduction}
The performance of online learning algorithms such as online gradient descent in adversarial, adaptive settings is a classic staple of optimization and game theory, e.g,  \cite{Cesa06,Fudenberg98,young2004strategic}. Arguably, the most well known results in this space
are the following:

\begin{itemize}
\item[i)] Sublinear regret of $O(\sqrt{T})$ is achievable in adversarial settings but only after employing a carefully chosen sequence of shrinking step-sizes or if the time horizon is finite and known in advance and the fixed learning rate is selected accordingly.
\item[ii)] Sublinear regret algorithms ``converge" to Nash equilibria in zero-sum games.
\end{itemize}

Despite the well established nature of these results recent work has revealed some surprising insights that come to challenge the traditional ways of thinking in this area.  Specifically, in the case of zero-sum games what is referred to as ``convergence" to equilibrium, is the fact that when both agent apply regret-minimizing algorithms, both the time-average of the mixed strategy profiles as well as the utilities of the agents converge approximately to their Nash equilibrium values, where the approximation error can become arbitrarily close to zero by choosing a sufficiently small step-size.
Naturally, this statement does not imply that 
 the day-to-day behavior converges to equilibria. 
In fact, the actual realized behavior  is antithetical to convergence to equilibrium.
 \cite{BaileyEC18} showed that \textit{Nash equilibria are repelling in zero-sum games} for all follow-the-regularized-leader dynamics. As seen in Figure \ref{fig:Regret} the dynamics spiral outwards away from the equilibrium.

These novel insights about the geometry of learning dynamics in zero-sum games suggest a much richer and not well understood landscape of coupled strategic behaviors. They also raise the tantalizing possibility that we may be able to leverage this knowledge to prove tighter regret bounds in games. In fact, a series of recent papers has focused on  beating the ``black-box" regret bounds using a combination of tailored dynamics and adaptive step-sizes, e.g,  \cite{Daskalakis:2011:NNA:2133036.2133057,rakhlin2013optimization,Syrgkanis:2015:FCR:2969442.2969573,foster2016learning} but so far no new bounds have been proven for the classic setting of fixed learning rates. 
 Interestingly, \cite{foster2016learning} explicitly examine the case of fixed learning rates $\eta$ to show that learning achieves sublinear ``approximate regret" where the algorithm compares itself against $(1-\eta)$ times the performance of the best action with hindsight. In contrast, our aim is to show sublinear regret for fixed $\eta$  using the standard notion of regret.

 Intuitively, non-equilibration and more generally this emergent behavioral complexity seem like harbingers of bad news in terms of system performance as well as of significant analytical obstacles. This pessimism seems especially justified given recent results about the behavior of online dynamics with fixed step-sizes in other small games (e.g. two-by-two coordination/congestion games), where their behavior can be shown to become provably chaotic (\cite{palaiopanos2017multiplicative,Thip18}). Nevertheless, we show that we can leverage this geometric information to provide the first to our knowledge sublinear regret guarantees for online gradient descent with fixed step-size in games.
Instability of Nash equilibria is  not an obstacle, but in fact may be leveraged as a tool, for proving 
 low regret.
 
 \bigskip

%Before we discuss our results and techniques, we will briefly touch upon the importance of the result. 

\medskip
{\bf Our results.} 
We  study the dynamics of gradient descent with fixed step size in two-strategy, two-player games. We leverage a deep understanding of the geometry of its orbits to prove the first sublinear regret bounds despite the constant learning rate. 
We show that the player strategies are repelled away from the Nash equilibrium. More specifically, regardless of the choice of the initial condition there are only a finite number of iterations where both players select mixed strategies (Theorem \ref{thm:boundary}). 
We prove a worst-case regret bound of $O(\sqrt{T})$ for arbitrarily learning without prior knowledge of $T$ (Theorem \ref{thm:regret}) matching the well known optimal bound for adaptive learning rates. 
An immediate corollary of our results is that  time-average of the mixed strategy profiles as well as the utilities of the agents converge to their  \textit{exact} Nash equilibrium values (and not to approximations thereof) (Corollary \ref{cor:Convergence}). Finally, we present a matching lower bound of $\Omega(\sqrt{T})$  (Theorem \ref{lower_bound}) establishing that our regret analysis is tight. % cannot be improved.
%{\color{red} improving on the best known results that imply convergence to an approximate Nash equilibrium.}

To obtain the upper bound, we establish a tight understanding of the geometry of the trajectories in the dual space, i.e., the trajectories of the payoff vectors. 
We show there exists a linear transformation of the payoff vectors that rotate around the Nash equilibrium.  
Moreover, the distance between the Nash equilibrium and these transformed utility vectors increases by a constant in each rotation (Lemma \ref{lem:LinearRadius}).
In addition, the time to complete a rotation is proportional to the distance between the Nash equilibrium and the transformed payoff vectors (Lemma \ref{lem:LinearSteps}).
Together, these results
imply a quadratic relationship between the number of iterations and the number of rotations completed establishing the $O(\sqrt{T})$ regret bound. We establish the lower bound by exactly tracking the strategies and regret for a single game.

%%%%%%%%%%%%%%%%%
\section{Preliminaries}\label{sec:prelim}

A two-player game consists of two players $\{1,2\}$ where each player has $n_i$ strategies to select from.  
Player $i$ can either select a pure strategy $j\in [n_i]$ or a mixed strategy $x_i\in {\cal X}_i=\{x_i\in \mathbb{R}^{n_i}_{\geq 0}: \sum_{j\in [n_i]} x_{ij}=1\}$.  A strategy is fully mixed if $x_i \in \mathbb{R}^{n_i}_{> 0}$.  

The most commonly studied class of games is zero-sum games.  
In a zero-sum game, there is a payoff matrix $A\in\mathbb{R}^{n_1\times n_2}$ where player 1 receives utility $x_1\cdot Ax_2$ and player $2$ receives utility $-x_1\cdot Ax_2$ resulting in the following optimization problem:
\begin{align}
	\max_{x_1\in {\cal X}_1}\min_{x_2\in {\cal X}_2}x_1\cdot Ax_2 \tag{Two-Player Zero-Sum Game}
\end{align}
The solution to this saddle problem  is the Nash equilibrium $x^{NE}$.  
If player 1 selects her Nash equilibria $x^{NE}_1$, then she guarantees her utility is $x^{NE}_1\cdot A x_2\geq x^{NE}_1\cdot A x_2^{NE}$ independent of what strategy player $2$ selects. 
$x^{NE}_1\cdot A x_2^{NE}$ is referred to as the value of the game.

\subsection{Online Learning in Continuous Time}
In many applications of game theory, players know neither the payoff matrix nor the Nash equilibria.
In such settings, players select their strategies adaptively. 
The most common way to do this in continuous time is by using a follow-the-regularized-leader (FTRL) algorithm. 
Given a strongly convex regularizer, a learning rate $\eta$, and an initial payoff vector $y_i(0)$, players select their strategies at time $T$ according to 
\begin{align}
y_1(T)&=y_1(0)+\int_{0}^TAx_2(t)dt \tag{Player 1 Payoff Vector}\\
y_2(T)&=y_2(0)-\int_{0}^TA^\intercal x_1(t)dt \tag{Player 2 Payoff Vector}\\
\tag{Continuous FTRL}\label{eqn:contFTRL} x_i(T)&=\argmax_{x_i\geq 0: \sum_{j\in [n_i]} x_{ij}=1}	\left\{y_i(T)\cdot x_i - \frac{h_i(x_i)}{\eta}\right\}
\end{align}
In this paper, we are primarily interested in the regularizer $h_i(x_i)=||x_i||_2^2/2$ resulting in the Gradient Descent algorithm: 
\begin{align}
\tag{Continuous Gradient Descent}\label{eqn:contGD}
x_i(t)&=\argmax_{x_i\geq 0: \sum_{j\in [n_i]} x_{ij}=1}	\left\{ y_i(t)\cdot x_i - \frac{||x_i||_2^2}{2\eta} \right\}
\end{align}
Continuous time FTRL learning in games has an interesting number of properties including time-average converge to the set of coarse correlated equilibria at a rate of $O(1/T)$  in general games \cite{GeorgiosSODA18} and thus to Nash
equilibria in zero-sum games. These systems can also exhibit interesting recurrent behavior e.g. periodicity \cite{DBLP:journals/corr/abs-1711-06879,nagarajan2018three}, Poincar\'{e} recurrence
\cite{GeorgiosSODA18,piliouras2014optimization,PiliourasAAMAS2014} and limit cycles \cite{paperics11}. These systems have formal  connections to Hamiltonian dynamics (i.e. energy perserving systems)  \cite{2019arXiv190301720B}. All of these types of recurrent behavior are special cases of chain recurrence  \cite{Entropy18,omidshafiei2019alpha}.

\subsection{Online Learning in Discrete Time}
In most settings, players update their strategies iteratively in discrete time steps.    
The most common class of online learning algorithms is again the family of follow-the-regularized-leader algorithms. 
\begin{align}
y_1^T&=y_1^0+\sum_{t=1}^{T-1}Ax_2^t \tag{Player 1 Payoff Vector}\\
y_2^T&=y_2^0-\sum_{t=1}^{T-1}A^\intercal x_1^t \tag{Player 2 Payoff Vector}\\
\tag{FTRL}\label{eqn:FTRL} x_i^t&=\argmax_{x_i\geq 0: \sum_{j\in [n_i]} x_{ij}=1}	\left\{y_i^t\cdot x_i - \frac{h_i(x_i)}{\eta}\right\}\\
\tag{Gradient Descent}\label{eqn:GD}
x_i^t&=\argmax_{x_i\geq 0: \sum_{j\in [n_i]} x_{ij}=1}	\left\{ y_i^t\cdot x_i - \frac{||x_i||_2^2}{2\eta} \right\}
\end{align}
where $\eta$ corresponds to the learning rate.
In Lemma \ref{lem:approx} of Appendix \ref{app:approx}, we show (\ref{eqn:FTRL}) is the first order approximation of (\ref{eqn:contFTRL}).

These algorithms again have interesting properties in zero-sum games. The time-average strategy converges to a $O(\eta)$-approximate Nash equilibrium \cite{Cesa06}. On the contrary,  Bailey and Piliouras show that the day-to-day behavior diverges away from interior Nash equilibria \cite{BaileyEC18}. For notational simplicity we do not introduce different learning rates $\eta_1, \eta_2$ but all of our proofs immediately carry over to this setting.

\subsection{Regret in Online Learning}
The most common way of analyzing an online learning algorithm is by examining its regret. 
The regret at time/iteration $T$ is the difference between the accumulated utility gained by the algorithm and the total utility of the best fixed action with hindsight.  
Formally for player 1,
\begin{align}
	Regret_1(T)&= \max_{x_1\in {\cal X}_1}\left\{\int_{0}^T x_1\cdot Ax_2(t)dt\right\}-\int_{0}^T x_1(t)\cdot Ax_2(t)dt \\
	Regret_1(T)&= \max_{x_1\in {\cal X}_1}\left\{\sum_{t=0}^T x_1\cdot Ax_2^t\right\}-\sum_{t=0}^T x_1^t\cdot Ax_2^t 
\end{align}
for continuous and discrete time respectively.

In the case of  (\ref{eqn:contFTRL}) %due to the stable nature of the dynamics
 it is possible to show rather strong regret guarantees. Specifically, \cite{GeorgiosSODA18} establishes that $Regret_1(T)\in O(1)$ even for non-zero-sum games. In contrast,  (\ref{eqn:FTRL}) only guarantees $Regret_1(T)\in O(\eta\cdot T)$ for a fixed learning rate.  In this paper, we utilize the geometry of  (\ref{eqn:GD}) to show $Regret_1(T)\in O(\sqrt{ T})$ in 2x2 zero-sum games ($n_1=n_2=2$).
%%%%%%%%%%%%%%%%%%%%%%%%%%%%%%
\section{The Geometry of \ref{eqn:GD}}\label{sec:geometry}
\begin{thm} Let $A$ be a \emph{2x2} game that has a unique fully mixed Nash equilibrium where strategies are updated according to (\ref{eqn:GD}). For any non-equilibrium initial strategies, there exists a $B$ such that $x^t$ is on the boundary for all $t\geq B$.\label{thm:boundary}
\end{thm}

Theorem \ref{thm:boundary} strengthens the result for (\ref{eqn:GD}) in 2x2 games from \cite{BaileyEC18}.  
Specifically, \cite{BaileyEC18} show that strategies come arbitrarily close to the boundary infinitely often when updated with any version of (\ref{eqn:FTRL}).  
This is accomplished by closely studying the geometry of the player strategies. 
We strengthen this result for (\ref{eqn:GD}) in 2x2 games by focusing on the geometry of the payoff vectors. 
The proof of Theorem \ref{thm:boundary} relies on many of the tools developed in Section \ref{sec:regret} for Theorem \ref{thm:regret} and is deferred to Appendix \ref{app:boundary}. \\
%%%%%%%%%%%%%%
%\subsection{KKT Optimality Conditions for \ref{eqn:GD}}
%\input{kkt.tex}
The first step to understanding the trajectories of the dynamics of (\ref{eqn:GD}), is characterizing the solution to (\ref{eqn:GD}). 
To streamline the discussion and presentation of results, we defer the proof of Lemma \ref{lem:solution} to Appendix \ref{app:kkt}. 

\begin{lem}\label{lem:solution}The solution to (\ref{eqn:GD}) is given by
	\begin{align}
	x_{ij}^t&= 
	\begin{cases} 
	0 		& \mbox{for } j \notin S_i \\
	\eta\left(y_{ij}^t-\sum_{k\in S_i} \frac{y_{ik}^t}{|S_i|}\right) +\frac{1}{|S_i|}& \mbox{for } j \in S_i
	\end{cases}.\label{eqn:solution}
	\end{align}
	where $S_i$ is found using Algorithm \ref{alg:search}.
\end{lem}

\begin{algorithm}
	\caption{Finding Optimal Set $S_i$}\label{alg:search}
	\begin{algorithmic}[1]
		\Procedure{Find $S_i$}{}
		\State $S_i \gets [n_i]$
		\BState	\emph{Search:}
		\State 	Select $j\in \argmin_{k\in S_i} \{y_{ik}^t\}$
		\State \textbf{if}		{$\eta\left(y_{ij}^t-\sum_{k\in S_i}\frac{y_{ik}^t}{|S_i|}\right)+\frac{1}{|S_i|}<0$}
			\State \hspace{.25in}	$S_i\gets S_i\setminus\{j\}$
			\State \hspace{.25in}	\textbf{goto} \emph{Search}
		\State \textbf{else}
			\State	\hspace{.25in} \textbf{return} $S_i$
		\EndProcedure
	\end{algorithmic}
\end{algorithm}

%%%%%%%%%%%%%
\subsection{Convex Conjugate of the Regularizer}
Our analysis primarily takes place in the space of payoff vectors. 
The payoff vector $y_i^t$ is a formal dual of the strategy $x_i^t$ obtained via
\begin{align}
h^*(y_i^t)=\max_{x_i\geq 0: \sum_{j\in [n_i]}x_{ij}=1} \left\{y_i^t\cdot x_i - \frac{h_i(x_i)}{\eta} \right\}
\end{align}
which is known as the convex conjugate or Fenchel Coupling of $h_i$ and is closely related to the Bregman Divergence. 
In \cite{GeorgiosSODA18} it is shown that the ``energy'' $r=\sum_{i=1}^2 h^*_i(y_i^t)$ is conserved in (\ref{eqn:contFTRL}).
By Lemma \ref{lem:approx}, (\ref{eqn:FTRL}) is the first order approximation of (\ref{eqn:contFTRL}).  The energy $\{y: r\leq \sum_{i=1}^2 h^*_i(y_i)\}$ is convex, and therefore the energy will be non-decreasing in (\ref{eqn:FTRL}). \cite{BaileyEC18} capitalized on this non-decreasing energy to show that strategies come arbitrarily close to the boundary infinitely often in (\ref{eqn:FTRL}). 

In a similar fashion, we precisely compute $h^*(y_i^t)$ to better understand the dynamics of (\ref{eqn:GD}). 
We deviate slightly from traditional analysis of (\ref{eqn:FTRL}) and embed the learning rate $\eta$ into the regularizer $h_i(x_i^t)$.  Formally, define $h_i(x_i^t)=||x_i^t||_2^2/(2\eta)$. Through the maximizing argument \cite{Kakade09}, we have
\begin{align}
	h^*_i(y_i^t)&= y_i^t\cdot x_i^t-\frac{||x_i^t||_2^2}{2\eta}.
\end{align} 
From Lemma \ref{lem:solution},
\begin{align} 
y_i^t\cdot x_i^t&=\sum_{j\in S_i} y_{ij}^t\left(\eta\left(y_{ij}^t-\sum_{k\in S_i} \frac{y_{ik}^t}{|S_i|}\right) +\frac{1}{|S_i|}\right)\\
&=\eta\sum_{j\in S_i} (y_{ij}^t)^2-\eta \sum_{j\in S_i}\sum_{k\in S_i}\frac{y_{ij}^ty_{ik}^t}{|S_i|}+\sum_{j\in S_i}\frac{y_{ij}^t}{|S_i|}
\end{align}
and
\begin{align}
\frac{||x_i^t||_2^2}{2\eta}&=\sum_{j\in S_i} \frac{\left(\eta\left(y_{ij}^t-\sum_{k\in S_i} \frac{y_{ik}^t}{|S_i|}\right) +\frac{1}{|S_i|}\right)^2}{2\eta}\\
%&=\frac{\eta}{2}\sum_{j\in S_i}(y_{ij}^t)^2-\eta\sum_{j\in S_i}\sum_{k\in S_i}\frac{y_{ij}^ty_{ij}^t}{|S_i|}+\sum_{j\in S_i}\frac{y_{ij}^t}{|S_i|}+\frac{\eta}{2}\frac{\left(\sum_{j\in S_i}y_{ij}^t\right)^2}{|S_i|}-\sum_{j\in S_i}\frac{y_{ij}^t}{|S_i|}+\frac{1}{2\eta}\frac{1}{|S_i|}\\
&=\frac{\eta}{2}\sum_{j\in S_i}(y_{ij}^t)^2-\eta\sum_{j\in S_i}\sum_{k\in S_i}\frac{y_{ij}^ty_{ij}^t}{|S_i|}+\frac{\eta}{2}\frac{\left(\sum_{j\in S_i}y_{ij}^t\right)^2}{|S_i|}+\frac{1}{2\eta}\frac{1}{|S_i|}.
\end{align}
Therefore, 
\begin{align}
h^*(y_i^t)_i&= y_i^t\cdot x_i^t-\frac{||x_i^t||_2^2}{2\eta}\\
&=\frac{\eta}{2}\sum_{j\in S_i}(y_{ij}^t)^2+\sum_{j\in S_i}\frac{y_{ij}^t}{|S_i|}-\frac{\eta}{2}\frac{\left(\sum_{j\in S_i}y_{ij}^t\right)^2}{|S_i|}-\frac{1}{2\eta}\frac{1}{|S_i|}.\label{eqn:notconvex}
\end{align}
%%%%%%%%%%%
\subsection{Selecting the Right Dual Space in 2x2 Games}\label{sec:dual}
Since $h_i(x_i)=||x_i||_2^2/(2\eta)$ is a strongly smooth function in the simplex, we expect for $h^*_i(y_i)$ to be strongly convex \cite{Kakade09} -- at least when it's corresponding dual variable $x_i$ is positive. However, (\ref{eqn:notconvex}) is not strongly convex for all $y_i^t\in \mathbb{R}^{n_i}$.  This is because $y_i^{t+1}$ cannot appear anywhere in $\mathbb{R}^{n_i}$.  Rather, $y_i^{t+1}$ is contained to a space ${\cal X}^*_i$ dual to the domain $\{x_i\in \mathbb{R}^{n_i}_{\geq 0}: \sum_{j=1}^{n_i} x_{ij}=1\}$. 

There are many non-intersecting dual spaces for the payoff vectors that  yield strategies $\{x_i^t\}_{t=1}^\infty$.  
\cite{GeorgiosSODA18} informally define a dual space  when they focus the analysis on the vector $y_i(t)-y_{in_i}(t)\mathbf{1}$.  Similarly, we define a dual space that will be convenient for showing our results in 2x2 zero-sum games. Consider the payoff matrix
\begin{align}
A=\left[\begin{array}{c c}
a & b \\
c & d \\ 
\end{array}\right]&\hspace{.2in}
\end{align}

Without loss of generality, we may assume $a>\min\{0,b,c\}$, $d>\min\{0,b,c\}$, and $A$ is singular, i.e., $ad-bc=0$ (see Appendix \ref{app:singularity} for details). Denote $\Delta y_1^t$ as 
\begin{align}
\Delta y_1^t	&= y_1^{t+1}-y_1^t\\
&= Ax_{2}^t\\
&=\left[\begin{array}{c}
(a-b)x_{21}^t+b\\
(c-d)x_{21}^t+d
\end{array}\right]
\end{align}
Therefore
\begin{align}
[d-c, a-b]\cdot \Delta y_1^t=ad-bc=0
\end{align}
since $A$ is singular. When $y_{11}^t$ increases by $a-b$, $y_{12}^t$ increases by $c-d$.  Thus, the vector $[a-b,c-d]$ describes the span of  the dual space ${\cal X}_1^*$. 
Moreover, (\ref{eqn:FTRL}) is invariant to constant shifts in the payoff vector $y_{1}^t$ and therefore we may assume $[d-c, a-b] \cdot y_1^0=0$.  By induction, 
\begin{align}
[d-c, a-b]\cdot y_1^t&= 	
[d-c, a-b]\cdot (y_1^{t-1} +\Delta y_1^{t-1})\\
&= 	
[d-c, a-b]\cdot y_1^{t-1}=0
\end{align}
This conveniently allows us to express $y_{12}^t$ in terms of $y_{11}^t$, 
\begin{align}
y_{12}^t&= \frac{c-d}{a-b}y_{11}^t.
\end{align}
Symmetrically, 
\begin{align}
y_{22}^t&= \frac{b-d}{a-c}y_{21}^t.
\end{align}
Combining these relationships with Lemma \ref{lem:solution} yields
\begin{align}
x_{11}^t	&=\begin{cases}
0 & \mbox{if } \eta\left(1-\frac{c-d}{a-b}\right)\frac{y_{11}^t}{2}+\frac{1}{2}\leq0\\
1 & \mbox{if } \eta\left(1-\frac{c-d}{a-b}\right)\frac{y_{11}^t}{2}+\frac{1}{2}\geq1\\
\eta\left(1-\frac{c-d}{a-b}\right)\frac{y_{11}^t}{2}+\frac{1}{2}			& \mbox{otherwise}
\end{cases}\\
x_{21}^t	&=\begin{cases}
0 & \mbox{if } \eta\left(1-\frac{b-d}{a-c}\right)\frac{y_{21}^t}{2}+\frac{1}{2}\leq0\\
1 & \mbox{if } \eta\left(1-\frac{b-d}{a-c}\right)\frac{y_{21}^t}{2}+\frac{1}{2}\geq1\\
\eta\left(1-\frac{b-d}{a-c}\right)\frac{y_{21}^t}{2}+\frac{1}{2}			& \mbox{otherwise}
\end{cases}
\end{align} 
The selection of this dual space also allows us to employ a convenient variable substitution to plot $x^t$ and $y^t$ on the same graph. 
\begin{align}
	z_1^t&= \eta \left( 1- \frac{c-d}{a-b}\right) \frac{y_{11}^t}{2}+\frac{1}{2}\\
	z_2^t&= \eta \left( 1- \frac{b-d}{a-c}\right) \frac{y_{21}^t}{2}+\frac{1}{2}
\end{align}
The strategy $x^t$ can now be expressed as 
\begin{align}
x_{i1}^t&=
	\begin{cases}
		0 & \mbox{ if } z_{i}^t\leq0\\
		1 & \mbox{ if } z_{i}^t\geq1\\ 
				z_{i}^t & \mbox{ otherwise} 
	\end{cases}
\end{align}
Moreover, (\ref{eqn:notconvex}) can be rewritten as
\begin{align}
h^*_i(y_{1}^t)=\bar{h}^*_1(z_1^t)&=\begin{cases}
\alpha_{10}z_1^t-\beta_{10} & \mbox{if } z_{1}^t\leq0\\
\alpha_{11}z_1^t-\beta_{11} & \mbox{if } z_{1}^t\geq1\\
\gamma_1(z_1^t)^2 + \alpha_{1}z_1^t-\beta_{1}			& \mbox{otherwise}\\
\end{cases}\\
h^*_i(y_{2}^t)=\bar{h}^*_2(z_2^t)
&=\begin{cases}
\alpha_{20}z_2^t-\beta_{20} & \mbox{if } z_{2}^t\leq0\\
\alpha_{21}z_2^t-\beta_{21} & \mbox{if } z_{2}^t\geq1\\
\gamma_2(z_2^t)^2 + \alpha_{2}z_2^t-\beta_{2}			& \mbox{otherwise}\\
\end{cases}
\end{align}
where $\alpha_{i0}< 0, \alpha_{i1}> 0,$ and $\gamma_i>0$. 
Both of these expressions are obviously strongly convex when the corresponding player strategy is in $(0,1)$. 
The full details of these reduction can be found in Appendix \ref{app:dualapp}. With this notation, $(x_{11}^t,x_{21}^t)$ is simply the projection of $z^t$ onto the unit square as shown in Figure \ref{fig:discrete}.

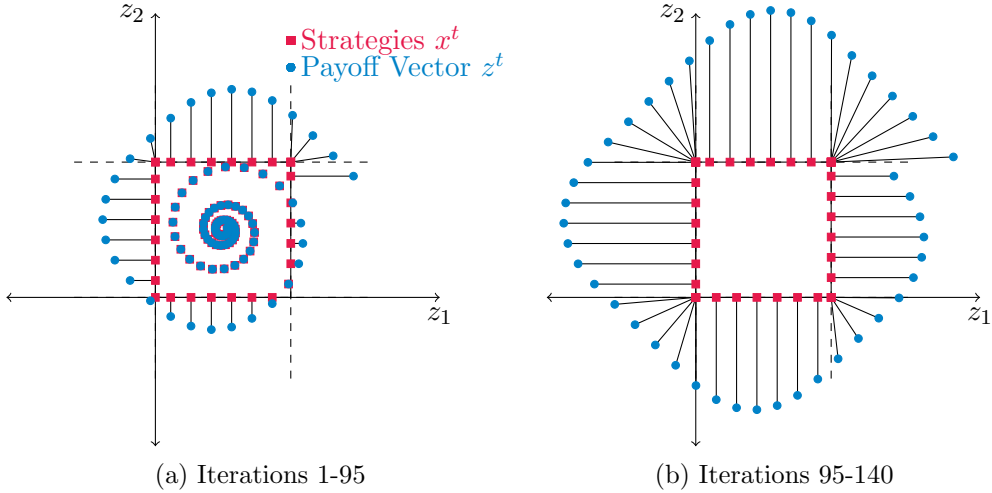
\begin{figure}[h]
	\begin{tabular}{c c}
		\begin{tikzpicture}[scale=1.8]
			\def\shift{.6}
			\def\axisshift{.5}
			\def\ETA{.5}
			\draw[<->] (0,0-\shift-\axisshift)--(0,1+\shift+\axisshift);
			\draw[<->] (0-\shift-\axisshift,0)--(1+\shift+\axisshift,0);
			\node[below] at (1+\shift+\axisshift,0) {$z_1$};
			\node[left] at (0,1+\shift+\axisshift) {$z_2$}; 
			\draw(1,0)--(1,1)--(0,1,0)--(0,0)--cycle;
			\draw[dashed] (0,-\shift)--(0,1+\shift);
			\draw[dashed] (1,-\shift)--(1,1+\shift);
			\draw[dashed] (-\shift,0)--(1+\shift,0);
			\draw[dashed] (-\shift,1)--(1+\shift,1);
			\draw (0.53,0.515)--(0.53,0.515);
\draw (0.5345,0.506)--(0.5345,0.506);
\draw (0.5363,0.49565)--(0.5363,0.49565);
\draw (0.534995,0.48476)--(0.534995,0.48476);
\draw (0.530423,0.4742615)--(0.530423,0.4742615);
\draw (0.52270145,0.4651346)--(0.52270145,0.4651346);
\draw (0.51224183,0.45832417)--(0.51224183,0.458324165);
\draw (0.49973908,0.45465162)--(0.49973908,0.454651616);
\draw (0.48613456,0.45472989)--(0.48613456,0.454729892);
\draw (0.47255353,0.45888952)--(0.47255353,0.458889523);
\draw (0.46022039,0.46712346)--(0.46022039,0.467123463);
\draw (0.45035743,0.47905735)--(0.45035743,0.479057347);
\draw (0.44407463,0.49395012)--(0.44407463,0.493950118);
\draw (0.44225967,0.51072773)--(0.44225967,0.510727729);
\draw (0.44547799,0.52804983)--(0.44547799,0.528049829);
\draw (0.45389293,0.54440643)--(0.45389293,0.544406433);
\draw (0.46721486,0.55823855)--(0.46721486,0.558238552);
\draw (0.48468643,0.56807409)--(0.48468643,0.568074093);
\draw (0.50510866,0.57266816)--(0.50510866,0.572668164);
\draw (0.52690911,0.57113557)--(0.52690911,0.571135567);
\draw (0.54824978,0.56306283)--(0.54824978,0.563062835);
\draw (0.56716863,0.5485879)--(0.56716863,0.548587901);
\draw (0.581745,0.52843731)--(0.581745,0.528437313);
\draw (0.59027619,0.50391381)--(0.59027619,0.503913814);
\draw (0.59145034,0.47683096)--(0.59145034,0.476830956);
\draw (0.58449962,0.44939586)--(0.58449962,0.449395855);
\draw (0.56931838,0.42404597)--(0.56931838,0.424045968);
\draw (0.54653217,0.40325045)--(0.54653217,0.403250455);
\draw (0.51750731,0.3892908)--(0.51750731,0.389290804);
\draw (0.48429455,0.38403861)--(0.48429455,0.384038612);
\draw (0.44950613,0.38875025)--(0.44950613,0.388750247);
\draw (0.41613121,0.40389841)--(0.41613121,0.403898408);
\draw (0.38730073,0.42905905)--(0.38730073,0.429059047);
\draw (0.36601844,0.46286883)--(0.36601844,0.462868828);
\draw (0.35487909,0.5030633)--(0.35487909,0.503063296);
\draw (0.35579808,0.54659957)--(0.35579808,0.546599569);
\draw (0.36977795,0.58986015)--(0.36977795,0.589860145);
\draw (0.39673599,0.62892676)--(0.39673599,0.62892676);
\draw (0.43541402,0.65990596)--(0.43541402,0.659905963);
\draw (0.48338581,0.67928176)--(0.48338581,0.679281756);
\draw (0.53717034,0.68426601)--(0.53717034,0.684266013);
\draw (0.59245014,0.67311491)--(0.59245014,0.673114912);
\draw (0.64438461,0.64537987)--(0.64438461,0.64537987);
\draw (0.68799858,0.60206449)--(0.68799858,0.602064486);
\draw (0.71861792,0.54566491)--(0.71861792,0.545664913);
\draw (0.73231739,0.48007954)--(0.73231739,0.480079537);
\draw (0.72634126,0.41038432)--(0.72634126,0.410384318);
\draw (0.69945655,0.34248194)--(0.69945655,0.342481941);
\draw (0.65220113,0.28264498)--(0.65220113,0.282644976);
\draw (0.58699463,0.23698464)--(0.58699463,0.236984636);
\draw (0.50809002,0.21088625)--(0.50809002,0.210886248);
\draw (0.42135589,0.20845924)--(0.42135589,0.208459243);
\draw (0.33389366,0.23205248)--(0.33389366,0.232052475);
\draw (0.25350941,0.28188438)--(0.25350941,0.281884376);
\draw (0.18807472,0.35583155)--(0.18807472,0.355831554);
\draw (0.14482419,0.44940914)--(0.14482419,0.449409138);
\draw (0.12964693,0.55596188)--(0.12964693,0.555961882);
\draw (0.14643549,0.6670678)--(0.14643549,0.667067804);
\draw (0.19655583,0.77313716)--(0.19655583,0.773137156);
\draw (0.27849698,0.86417041)--(0.27849698,0.864170406);
\draw (0.3877481,0.93062131)--(0.3877481,0.930621312);
\draw (0.5169345,0.96429688)--(0.5169345,0.964296882);
\draw (0.65622356,0.95921653)--(0.65622356,0.959216533);
\draw (0.79398852,0.91234947)--(0.79398852,0.912349465);
\draw (0.91769336,0.82415291)--(0.91769336,0.82415291);
\draw (1,0.6988449)--(1.01493923,0.698844902);
\draw (1,0.5488449)--(1.0745927,0.548844902);
\draw (1,0.3988449)--(1.08924617,0.398844902);
\draw (1,0.2488449)--(1.05889964,0.248844902);
\draw (0.98355311,0.0988449)--(0.98355311,0.098844902);
\draw (0.86320658,0)--(0.86320658,-0.046221033);
\draw (0.71320658,0)--(0.71320658,-0.155183008);
\draw (0.56320658,0)--(0.56320658,-0.219144984);
\draw (0.41320658,0)--(0.41320658,-0.238106959);
\draw (0.26320658,0)--(0.26320658,-0.212068934);
\draw (0.11320658,0)--(0.11320658,-0.14103091);
\draw (0,0)--(-0.03679342,-0.024992885);
\draw (0,0.12500711)--(-0.18679342,0.125007115);
\draw (0,0.27500711)--(-0.29929128,0.275007115);
\draw (0,0.42500711)--(-0.36678915,0.425007115);
\draw (0,0.57500711)--(-0.38928701,0.575007115);
\draw (0,0.72500711)--(-0.36678488,0.725007115);
\draw (0,0.87500711)--(-0.29928274,0.875007115);
\draw (0,1)--(-0.18678061,1.025007115);
\draw (0,1)--(-0.03678061,1.175007115);
\draw (0.11321939,1)--(0.11321939,1.325007115);
\draw (0.26321939,1)--(0.26321939,1.441041297);
\draw (0.41321939,1)--(0.41321939,1.51207548);
\draw (0.56321939,1)--(0.56321939,1.538109663);
\draw (0.71321939,1)--(0.71321939,1.519143845);
\draw (0.86321939,1)--(0.86321939,1.455178028);
\draw (1,1)--(1.01321939,1.346212211);
\draw (1,1)--(1.16321939,1.196212211);
\draw (1,1)--(1.31321939,1.046212211);
\draw (1,0.89621221)--(1.46321939,0.896212211);
			\draw[COLOR1] plot[mark=square*, only marks,mark options={fill=COLOR1}, mark size=.8] (1,1.9);
			\node[right] at (1,1.9) {\color{COLOR1}Strategies $x^t$};
			\draw[COLOR2] plot[mark=*, only marks,mark options={fill=COLOR2}, mark size=.8] (1,1.7);
			\node[right] at (1,1.7) {\color{COLOR2}Payoff Vector $z^t$};
			%%%%%%%%%% Data Points
			\draw[COLOR1] plot[mark=square*, only marks,mark options={fill=COLOR1}, mark size=.8] file {strat95.txt};
			\draw[COLOR2] plot[mark=*, only marks,mark options={fill=COLOR2}, mark size=.8] file {payoff95.txt};
		\end{tikzpicture}
		&
		\begin{tikzpicture}[scale=1.8]
		\def\shift{.6}
		\def\axisshift{.5}
		\def\ETA{.5}
		\draw[<->] (0,0-\shift-\axisshift)--(0,1+\shift+\axisshift);
		\draw[<->] (0-\shift-\axisshift,0)--(1+\shift+\axisshift,0);
		\node[below] at (1+\shift+\axisshift,0) {$z_1$};
		\node[left] at (0,1+\shift+\axisshift) {$z_2$}; 
		\draw(1,0)--(1,1)--(0,1,0)--(0,0)--cycle;
		\draw[dashed] (0,-\shift)--(0,1+\shift);
		\draw[dashed] (1,-\shift)--(1,1+\shift);
		\draw[dashed] (-\shift,0)--(1+\shift,0);
		\draw[dashed] (-\shift,1)--(1+\shift,1);
		\draw (1,0.89621221)--(1.46321939,0.896212211);
\draw (1,0.74621221)--(1.58208305,0.746212211);
\draw (1,0.59621221)--(1.65594672,0.596212211);
\draw (1,0.44621221)--(1.68481038,0.446212211);
\draw (1,0.29621221)--(1.66867404,0.296212211);
\draw (1,0.14621221)--(1.60753771,0.146212211);
\draw (1,0)--(1.50140137,-0.003787789);
\draw (1,0)--(1.35140137,-0.153787789);
\draw (1,0)--(1.20140137,-0.303787789);
\draw (1,0)--(1.05140137,-0.453787789);
\draw (0.90140137,0)--(0.90140137,-0.603787789);
\draw (0.75140137,0)--(0.75140137,-0.724208201);
\draw (0.60140137,0)--(0.60140137,-0.799628612);
\draw (0.45140137,0)--(0.45140137,-0.830049023);
\draw (0.30140137,0)--(0.30140137,-0.815469434);
\draw (0.15140137,0)--(0.15140137,-0.755889845);
\draw (0.00140137,0)--(0.00140137,-0.651310256);
\draw (0,0)--(-0.14859863,-0.501730667);
\draw (0,0)--(-0.29859863,-0.351730667);
\draw (0,0)--(-0.44859863,-0.201730667);
\draw (0,0)--(-0.59859863,-0.051730667);
\draw (0,0.09826933)--(-0.74859863,0.098269333);
\draw (0,0.24826933)--(-0.86911783,0.248269333);
\draw (0,0.39826933)--(-0.94463703,0.398269333);
\draw (0,0.54826933)--(-0.97515623,0.548269333);
\draw (0,0.69826933)--(-0.96067543,0.698269333);
\draw (0,0.84826933)--(-0.90119463,0.848269333);
\draw (0,0.99826933)--(-0.79671383,0.998269333);
\draw (0,1)--(-0.64723303,1.148269333);
\draw (0,1)--(-0.49723303,1.298269333);
\draw (0,1)--(-0.34723303,1.448269333);
\draw (0,1)--(-0.19723303,1.598269333);
\draw (0,1)--(-0.04723303,1.748269333);
\draw (0.10276697,1)--(0.10276697,1.898269333);
\draw (0.25276697,1)--(0.25276697,2.017439242);
\draw (0.40276697,1)--(0.40276697,2.091609151);
\draw (0.55276697,1)--(0.55276697,2.12077906);
\draw (0.70276697,1)--(0.70276697,2.104948969);
\draw (0.85276697,1)--(0.85276697,2.044118878);
\draw (1,1)--(1.00276697,1.938288787);
\draw (1,1)--(1.15276697,1.788288787);
\draw (1,1)--(1.30276697,1.638288787);
\draw (1,1)--(1.45276697,1.488288787);
\draw (1,1)--(1.60276697,1.338288787);
\draw (1,1)--(1.75276697,1.188288787);
\draw (1,1)--(1.90276697,1.038288787);
		%%%%%%%%%% Data Points
		\draw[COLOR1] plot[mark=square*, only marks,mark options={fill=COLOR1}, mark size=.8] file {strat140.txt};
		\draw[COLOR2] plot[mark=*, only marks,mark options={fill=COLOR2}, mark size=.8] file {payoff140.txt};
		\end{tikzpicture}
		\\
		\small (a) Iterations 1-95
		&
		\small (b) Iterations 95-140
	\end{tabular}
\caption{Strategies and Transformed Payoff Vectors Rotating Clockwise and Outwards in Matching Pennies with $\eta=.15$ and $(y^0_{11},y^0_{11})=(.2,-.3)$.}\label{fig:discrete}
\end{figure}

%%%%%%%%%%
\section{$\Theta(\sqrt{T})$ Regret in 2x2 Zero-Sum Games}\label{sec:regret}

\begin{thm}
	Let $A$ be a \emph{2x2} game that has a unique fully mixed Nash equilibrium. When $x^t$ is updated according to (\ref{eqn:GD}), $Regret_1(T)\in  O\left(\sqrt{T}\right)$. \label{thm:regret}
\end{thm}

It is well known that if an algorithm admits sublinear regret in zero-sum games, then the time-average play converges to a Nash equilibirum.  Thus, Theorem \ref{thm:regret} immediately results in the following corollary. 

\begin{cor}\label{cor:Convergence}
		Let $A$ be a \emph{2x2} game that has a unique fully mixed Nash equilibrium. When $x^t$ is updated according to (\ref{eqn:GD}), average strategy $\bar{x}^T=\sum_{t=1}^T\frac{x^t}{T}$ converges to $x^{NE}$ as $T\to \infty$. 
\end{cor}

\begin{proof}[Proof of Theorem \ref{thm:regret}]	
	The result is simple if $x^1=x^{NE}$.  Neither player strategy will ever change.  Since player 1's opponent is playing the fully mixed $x_2^{NE}$, player 1's utility is constant independent of what strategy is selected and therefore the regret is always $0$. Now consider $x^1\neq x^{NE}$. 
	
	The main details of the proof are captured in Figure \ref{fig:ProofSketch}.
	Specifically in Section \ref{sec:partition}, we  establish break points $t_0<t_1<...<t_k=T+1$ and analyze the impact strategies $x^{t_j}, x^{t_j+1},...,x^{t_{j+1}-1}$ have on the regret. The strategies $x^{t_j}, x^{t_j+1},...,x^{t_{j+1}-1}$ are contained in adjacent red and green sections as shown in Figure \ref{fig:ProofSketch}.
	
	\begin{figure}\centering
		\begin{tikzpicture}[scale=1.8]
		\def\shift{.6}
		\def\w{.4}		
		\def\axisshift{.6}
		\def\ETA{.5}
		\draw[<->] (0,0-\shift-\axisshift)--(0,1+\shift+\axisshift);
		\draw[<->] (0-\shift-\axisshift,0)--(1+\shift+\axisshift,0);
		\node[below left] at (1+\shift+\axisshift,0) {$z_1$};
		\node[below left] at (0,1+\shift+\axisshift) {$z_2$}; 
		\draw(1,0)--(1,1)--(0,1,0)--(0,0)--cycle;
		\draw[ultra thick] (0,0-\shift-\axisshift)--(0,1+\shift+\axisshift);
		\draw[ultra thick] (1,0-\shift-\axisshift)--(1,1+\shift+\axisshift);
		\draw[ultra thick] (0-\shift-\axisshift,0)--(1+\shift+\axisshift,0);
		\draw[ultra thick] (0-\shift-\axisshift,1)--(1+\shift+\axisshift,1);
		
		\fill[red,opacity=.3] (0,0)--(1,0)--(1,0-\shift-\axisshift)--(0,0-\shift-\axisshift)--cycle;
		\fill[red,opacity=.3] (0,1)--(1,1)--(1,1+\shift+\axisshift)--(0,1+\shift+\axisshift)--cycle;
		\fill[red,opacity=.3] (0,0)--(0,1)--(0-\shift-\axisshift,1)--(0-\shift-\axisshift,0)--cycle;
		\fill[red,opacity=.3] (1,0)--(1,1)--(1+\shift+\axisshift,1)--(1+\shift+\axisshift,0)--cycle;
		
		\fill[red,opacity=.3] (2.5,1.4-\w)--(2.5,.6-\w)--(6.6,.6-\w)--(6.6,1.4-\w)--cycle;
		
		\node[right] at (2.5,1.2-\w) {Energy $r_j$ increases by $\Theta(1)$ per iteration.};
		
		\node[right] at (2.5,.8-\w) {There are $\Theta(1)$  iterations per rotation.};

		\fill[green,opacity=.3] (2.5,.4-\w)--(2.5,-.4-\w)--(6.6,-.4-\w)--(6.6,.4-\w)--cycle;
		
		\node[right] at (2.5,.2-\w) {Energy $r_j$ does not change per iteration.};
		
		\node[right] at (2.5,-.2-\w) {There are $\Theta(r_j)$ iterations per rotation.};

		\fill[green,opacity=.3] (0,0)--(0,0-\shift-\axisshift)--(0-\shift-\axisshift,0-\shift-\axisshift)--(0-\shift-\axisshift,0)--cycle;
		\fill[green,opacity=.3] (1,0)--(1,0-\shift-\axisshift)--(1+\shift+\axisshift,0-\shift-\axisshift)--(1+\shift+\axisshift,0)--cycle;
		\fill[green,opacity=.3] (0,1)--(0,1+\shift+\axisshift)--(0-\shift-\axisshift,1+\shift+\axisshift)--(0-\shift-\axisshift,1)--cycle;
		\fill[green,opacity=.3] (1,1)--(1,1+\shift+\axisshift)--(1+\shift+\axisshift,1+\shift+\axisshift)--(1+\shift+\axisshift,1)--cycle;
		
		\draw[COLOR2] plot[mark=*, only marks,mark options={fill=COLOR2}, mark size=.8] (2.6,1.8-\w);
		\node[right] at (2.7,1.8-\w)
		{\color{COLOR2} Payoff Vector $z^t$};

		\draw[COLOR1] plot[mark=square*, only marks,mark options={fill=COLOR1}, mark size=.8] (2.6,1.8);
		\node[right] at (2.7,1.8) {\color{COLOR1}Strategies $x^t$};

		%%%%%%%%%% Data Points
		\draw[COLOR2] plot[mark=*, only marks,mark options={fill=COLOR2}, mark size=.8] file {payoff140.txt};
				\draw (1,0.89621221)--(1.46321939,0.896212211);
\draw (1,0.74621221)--(1.58208305,0.746212211);
\draw (1,0.59621221)--(1.65594672,0.596212211);
\draw (1,0.44621221)--(1.68481038,0.446212211);
\draw (1,0.29621221)--(1.66867404,0.296212211);
\draw (1,0.14621221)--(1.60753771,0.146212211);
\draw (1,0)--(1.50140137,-0.003787789);
\draw (1,0)--(1.35140137,-0.153787789);
\draw (1,0)--(1.20140137,-0.303787789);
\draw (1,0)--(1.05140137,-0.453787789);
\draw (0.90140137,0)--(0.90140137,-0.603787789);
\draw (0.75140137,0)--(0.75140137,-0.724208201);
\draw (0.60140137,0)--(0.60140137,-0.799628612);
\draw (0.45140137,0)--(0.45140137,-0.830049023);
\draw (0.30140137,0)--(0.30140137,-0.815469434);
\draw (0.15140137,0)--(0.15140137,-0.755889845);
\draw (0.00140137,0)--(0.00140137,-0.651310256);
\draw (0,0)--(-0.14859863,-0.501730667);
\draw (0,0)--(-0.29859863,-0.351730667);
\draw (0,0)--(-0.44859863,-0.201730667);
\draw (0,0)--(-0.59859863,-0.051730667);
\draw (0,0.09826933)--(-0.74859863,0.098269333);
\draw (0,0.24826933)--(-0.86911783,0.248269333);
\draw (0,0.39826933)--(-0.94463703,0.398269333);
\draw (0,0.54826933)--(-0.97515623,0.548269333);
\draw (0,0.69826933)--(-0.96067543,0.698269333);
\draw (0,0.84826933)--(-0.90119463,0.848269333);
\draw (0,0.99826933)--(-0.79671383,0.998269333);
\draw (0,1)--(-0.64723303,1.148269333);
\draw (0,1)--(-0.49723303,1.298269333);
\draw (0,1)--(-0.34723303,1.448269333);
\draw (0,1)--(-0.19723303,1.598269333);
\draw (0,1)--(-0.04723303,1.748269333);
\draw (0.10276697,1)--(0.10276697,1.898269333);
\draw (0.25276697,1)--(0.25276697,2.017439242);
\draw (0.40276697,1)--(0.40276697,2.091609151);
\draw (0.55276697,1)--(0.55276697,2.12077906);
\draw (0.70276697,1)--(0.70276697,2.104948969);
\draw (0.85276697,1)--(0.85276697,2.044118878);
\draw (1,1)--(1.00276697,1.938288787);
\draw (1,1)--(1.15276697,1.788288787);
\draw (1,1)--(1.30276697,1.638288787);
\draw (1,1)--(1.45276697,1.488288787);
\draw (1,1)--(1.60276697,1.338288787);
\draw (1,1)--(1.75276697,1.188288787);
\draw (1,1)--(1.90276697,1.038288787);
		%%%%%%%%%% Data Points
		\draw[COLOR1] plot[mark=square*, only marks,mark options={fill=COLOR1}, mark size=.8] file {strat140.txt};
		\end{tikzpicture}
		\caption{Partitioning of Payoff Vectors for the Proof of Theorem \ref{thm:regret}.}\label{fig:ProofSketch}
	\end{figure}
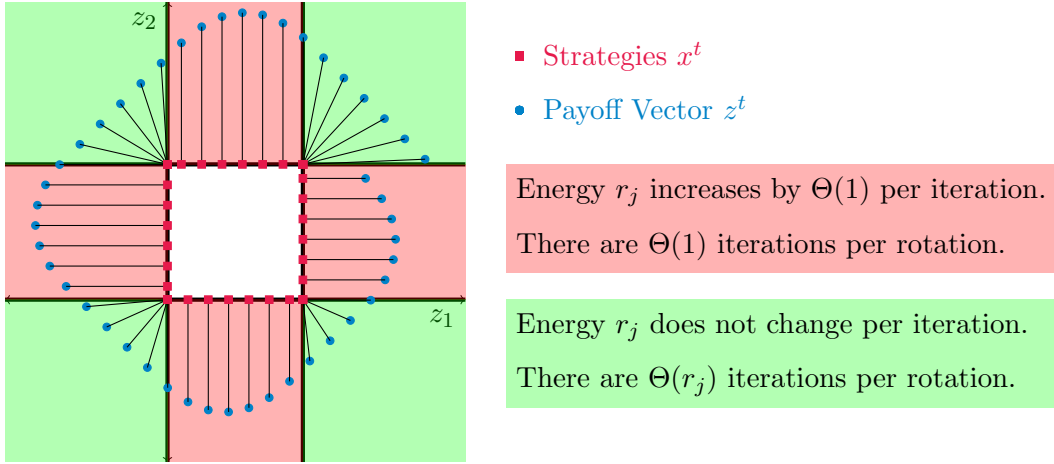

	Next in Section \ref{sec:constant}, we show  that there exists $\Theta(1)$ iterations where $x^t\neq x^{t+1}$ for each partitioning, $\{t_j,t_j+1,...,t_{j+1}-1\}$. 
	Specifically, we show that $\Theta(1)$ consecutive payoff vectors appear in a red section of Figure \ref{fig:ProofSketch}. 
	The remaining points all appear in a green section and the corresponding player strategies are equivalent. 
	This implies
		\begin{align}\sum_{t=t_j}^{t_{j+1}-1}(x_1^{t+1}-x_1^t)\cdot Ax_2^t
			&=\sum_{{t\in[t_j, t_{j+1}-1]:}{ x_1^{t+1}\neq x_1^t}}(x_1^{t+1}-x_1^t)\cdot Ax_2^t\\
			&\in \sum_{{t\in[t_j, t_{j+1}-1]:}{ x_1^{t+1}\neq x_1^t}}O(1)\\
			&\in O(1)
		\end{align}
		
	Denote $r_j=\sum_{i=1}^t \bar{h}_i^*(z_i^{t_j})$ as the total energy of the system in iteration $t_j$.  In Section \ref{sec:energy}, we show this energy increases linearly in each partition, i.e., $r_{j+1}-r_j\in \Theta(1)$.
	In Section \ref{sec:steps}, we also show that the size of each partition is proportional to the energy in the system at the beginning of that partition, i.e., $t_{j+1}-t_j\in \Theta(r_j)$. 
	Combining these two, $t_{j}\in \Theta(j^2)$. 
	Therefore $T\in \Theta(k^2)$ and $k\in \Theta\left(\sqrt{T}\right)$ where $k$ is the total number of partitions.
	Finally, it is well known (\cite{Cesa06}) that the regret of player $1$ in zero-sum games through $T$ iterations is bounded by
	\begin{align}
		Regret_1(T) &\leq O(1)+ \sum_{t=0}^T(x_1^{t+1}-x_1^t)\cdot Ax_2^t\\
					&\leq O(1)+ \sum_{t=0}^{t_0-1}(x_1^{t+1}-x_1^t)\cdot Ax_2^t +\sum_{i=1}^{k}\sum_{t=t_{i-1}}^{t_i-1}(x_1^{t+1}-x_1^t)\cdot Ax_2^t\label{eqn:partion}\\
					&\in O(1) +\sum_{i=1}^{k}O(1)\label{eqn:steps}\\
					&\in O\left(\sqrt{T}\right)\label{eqn:energy}
	\end{align}
	completing the proof of the theorem. 
\end{proof}

Next, we provide a game and initial conditions that has regret $\Theta(\sqrt{T})$ establishing that the bound in Theorem \ref{thm:regret} is tight. 

\begin{thm}\label{lower_bound}
	Consider the game Matching Pennies with learning rate $\eta=1$ and initial conditions $y_{1}^0=y_2^0=(1,0)$. Then player 1's regret is $\Theta(\sqrt{T})$ when strategies are updated with (\ref{eqn:GD}). 
\end{thm}

The proof follows similarly to the proof of Theorem \ref{thm:regret} by exactly computing the regret in every iteration of (\ref{eqn:GD}). 
The full details appear in Section \ref{sec:Lower}.
%input{Lower Bound.tex}

\section{Related Work}

The study of learning dynamics in game theory has a long history dating back work of \cite{Brown1951} and \cite{Robinson1951} on fictitious play in zero-sum games, which  followed shortly after 
von Neumann's seminal work on zero-sum games (\cite{Neumann1928,Neumann1944}).
Some good reference books are the following: \cite{Fudenberg98,Cesa06,young2004strategic}.
 The classic results about time-average convergence of no-regret dynamics have been successfully generalized to include multiplayer extensions of network constant-sum games by
\cite{dask09,Cai,cai2016zero}.

 \textit{Non-equilibrating dynamics in algorithmic game theory.}
%Although traditionally research in the area aims at proving results about convergence of dynamics to equilibria, 
In recent years the algorithmic game theory community has produced several  interesting non-equilibrium results. These proofs are typically based on ad-hoc techniques, and results in this area typically revolve around specific examples of games with a handful of agents and strategies.  
\cite{daskalakis10} show that multiplicative weights update (MWU) does not converge even in a time-average sense in the case of a specific 3x3 game.
 \cite{paperics11} establish non-convergence for a continuous-time variant of MWU, known as the replicator dynamic, for a 2x2x2 game and show that as a result the system social welfare converges to states that dominate all Nash equilibria. 
 %\cite{Balcan12} study MWU in rank-1 games (i.e. games where the summation of the payoff matrices of the two agents results in a rank-1 matrix). This games are in some sense almost constant-sum, and the paper shows that there exist such games where not even the time-average of MWU converges to equilibria. 
  \cite{palaiopanos2017multiplicative,Thip18} prove the existence of Li-Yorke chaos in MWU dynamics of 2x2 potential games. Our result add a new chapter in this area with new detailed understanding of the non-equilibrium trajectories of gradient descent in two-by-two zero-sum games and their implications to regret.
 
 \textit{Connections to continuous time dynamics in game theory.} 
 From the perspective of evolutionary game theory, which typically studies continuous time dynamics, numerous nonconvergence results are known but again typically for small games, e.g., \cite{sandholm10}. 
  \cite{piliouras2014optimization} and \cite{PiliourasAAMAS2014} show that replicator dynamics, the continuous time version of MWU exhibit a specific type of near periodic behavior, which is known as Poincar\'{e} recurrence.  
  Recently, \cite{GeorgiosSODA18} show  how to generalize this recurrent behavior for replicator to more general continuous time variants of FTRL dynamics.  \cite{2017arXiv171011249M} show that these arguments can also be adapted in the case of dynamically evolving games.
   %Proving periodicity instead of Poincar\'{e} recurrence hinges upon proving hidden planar nature of some of these systems. These may arise in the case of team competition (Boolean functions, sexual evolution between prey and predator), e.g,  
   Cycles arise also in team competition (\cite{DBLP:journals/corr/abs-1711-06879}) as well as in network competition (\cite{nagarajan2018three}). % or in triangle zero-sum network games \cite{}.
  The papers in this category combine  delicate arguments  such as volume preservation and the existence of  constants of motions (``energy preservation") for the dynamics to establish cyclic behavior.
   In the case of discrete time dynamics, such as the multiplicative weights or gradient descent, the system trajectories are first order approximations of the above motion  and these conservation arguments are no longer valid. Instead as we have seen in this paper the ``energy" is not preserved but increases over time at a predictable rate that allows us to prove tight bounds on the regret. 
Finally, \cite{Entropy18} have put forward a program for linking game theory and topology of dynamical systems. % that holds the promise of becoming a universal tool for analyzing non-equilibrium dynamics in games.

 \textit{Fast regret minimization in games.} 
 It is widely known that % the time-average of no-regret algorithms converge to the set of coarse correlated.
 the  ``black-box" average regret rate of $O(1/\sqrt{t})$  it is achieved by MWU with suitably shrinking step size without making any assumptions about its environment.
Recently, several authors have focused instead on obtaining stronger regret guarantees for systems of learning algorithms in games.
%The main idea is that in such systems the player's opponents are minimizing their own regret, rather than trying to maximize the player's regret and thus stronger guarantees could be achievable.
 \cite{Daskalakis:2011:NNA:2133036.2133057} and   \cite{rakhlin2013optimization} develop no-regret dynamics with a $O(\log t/t)$ regret minimization rate when played against each other in two-player zero-sum games.
\cite{Syrgkanis:2015:FCR:2969442.2969573} further analyze a recency biased variant of FTRL in more general  games and showed a $O(t^{-3/4})$ regret minimization rate.
The social welfare converges at a rate of $O(t^{-1})$, a result which was extended to standard versions of FRTL dynamics by \cite{foster2016learning}. 

\textit{Learning in zero-sum games and applications to Artificial Intelligence.}
%Generative Adversarial Networks (GANs)  is an AI technique used in unsupervised machine learning, implemented by a system of two neural networks contesting with each other in a zero-sum game framework (\cite{Goodfellow:2014:GAN:2969033.2969125}).
 A stream of recent papers proves positive results about convergence to equilibria in (mostly bilinear) zero-sum games for suitable adapted variants of first-order methods and then apply these techniques to Generative Adversarial Networks (GANs), showing improved performance (e.g. \cite{daskalakis2017training}) 
 %study the convergence properties of a specifically tailored dynamic in bilinear saddle problems and show convergence of the daily behavior to equilibrium. 
  \cite{2018arXiv180205642B} exploit conservation laws of learning dynamics in zero-sum games (e.g. \cite{piliouras2014optimization,GeorgiosSODA18}) to develop new algorithms for training GANs that add a new component to the dynamic that aims at minimizing this energy function. Different energy shrinking techniques for convergence in GANs (non-convex saddle point problems) exploit connections to variational inequalities and employ mirror descent techniques with an extra gradient step (\cite{2018arXiv180702629M}). 
 Game theoretic inspired techniques such as time-averaging seem to work well in practice for a wide range of architectures (\cite{2018arXiv180604498Y}). 
 
 Finally, the emergence of cycles in zero-sum competition lies at the core of some of the most exciting problems in creating artificial agents  for complex environments such as Starcraft, where even evaluating the strength of an individual agent is a non-trivial task (\cite{DBLP:journals/corr/abs-1806-02643}). Recent  approaches are inspired by the emergence of cyclic behavior to introduce algorithms that aim at game-theoretic niching (\cite{balduzzi2019open}). 
 %Once again, the emergence of cycling behavior can after careful consideration be exploited and used as a tool for AI system design.

\section{Conclusion}

We present the first, to our knowledge, proof of sublinear regret for the most classic FTRL dynamic, online gradient descent, in two-by-two zero-sum games.
 Our proof techniques leverage geometric information and hinge upon the fact that FTRL dynamics, although are typically referred to as ``converging" 
 to Nash equilibria in zero-sum games, diverge away from them. We strongly believe that these techniques, which we are just introducing, are far 
 from being fully mined. Although several novel ideas will be required, we are fairly confident that these sublinear regret bounds carry over to much more
 general classes of FTRL dynamics as well as to large (zero-sum) games. 
\section*{Acknowledgements}

James P. Bailey and Georgios Piliouras acknowledge SUTD grant SRG ESD 2015 097, MOE AcRF Tier 2 Grant 2016-T2-1-170,  grant PIE-SGP-AI-2018-01 and NRF 2018 Fellowship NRF-NRFF2018-07.
\bibliographystyle{acm}
\bibliography{IEEEabrv,Bibliography,refer}
\newpage
\appendix
%%%%%%%%%%%%%%%%%%%%%%%%%%%%%%%%%%%%%%%%%%%%%%%%%%%%%%%%%%%%%%%%%%%%%%%%%
\section{First Order Approximation of (\ref{eqn:contFTRL})}\label{app:approx}
\begin{lem}\label{lem:approx}
	(\ref{eqn:FTRL}) is the first order approximation of (\ref{eqn:contFTRL}).
\end{lem}
\begin{proof}The first order approximation of $y_1(t)$ is 
	\begin{align}
		\hat{y}_{1}(t)	&= 	{y}_1(t-1)+{\frac{d}{dt} y_1(t-1)}\\
						&= 	{y}_1(t-1)+A{x}_1(t-1)
	\end{align}
	and 
	\begin{align}
		\hat{x}_1(t)	&= \argmax_{x_1\in {\cal X}_1}\left\{x\cdot \hat{y}_1(t) - \frac{h_1(x_1)}{\eta}\right\}
	\end{align}
	Inductively, $\hat{y}_{1}(t)=y_1^t$  and $\hat{x}_1(t)=x_1^t$  as defined in (\ref{eqn:FTRL}) completing the proof of the lemma.
\end{proof}

%%%%%%%%%%%%%%%%%%%%%%%%%%%%%%%%%%%%%%%%%%%%%%%%%%%%%%%%%%%%%%%%%%%%%%%%%
\section{Optimal Solution to (\ref{eqn:GD})}\label{app:kkt}
The KKT optimality conditions (see \cite{bertsekas1999nonlinear}) for (\ref{eqn:GD}) are given by
\begin{align}
x_i^t&= \eta\left( y^t_i-\lambda_i^t\cdot \mathbf{1} + u_i^t\right)  \tag{Critical Point}\label{eqn:KKT1}\\
x_i^t&\geq 0 \tag{Non-negativity} \\
\sum_{j=1}^{n_i} x_{ij}^t &=1 \tag{Primal Feasibility}\label{eqn:KKT3}\\
u_i^t& \geq 0 \tag{Dual Feasibility}\\
u_i^t\cdot x_i^t &= 0 \tag{Complimentary Slackness}\label{eqn:CS} 
\end{align}
where $u_i^t\in \mathbb{R}^{n_i}$ and $\lambda_i^t\in \mathbb{R}$. 

Let $S_i$ be the set of $j$ where $u_{ij}^t=0$. By (\ref{eqn:CS}), $x_{ij}^t=0$ for all $j \notin S_i$. Therefore, (\ref{eqn:KKT1}) becomes
\begin{align}
x_{ij}^t&= 
\begin{cases} 
0 		& \mbox{for } j \notin S_i \\
\eta(y_{ij}^t-\lambda_i^t) & \mbox{for } j \in S_i
\end{cases}.\label{eqn:cases}
\end{align}  
Substituting (\ref{eqn:cases}) into (\ref{eqn:KKT3}) yields
\begin{align}
1	&= \sum_{j=1}^{n_i} x_{ij}^t\\
&= \sum_{j\in S_i} \eta(y_{ij}^t-\lambda_i^t)
\end{align}
and $\lambda_i^t= \sum_{j\in S_i} y_{ij}^t/|S_i|-1/(\eta |S_i|)$.  Therefore
\begin{align}
x_{ij}^t&= 
\begin{cases} 
0 		& \mbox{for } j \notin S_i \\
\eta\left(y_{ij}^t-\sum_{k\in S_i} \frac{y_{ik}^t}{|S_i|}\right) +\frac{1}{|S_i|}& \mbox{for } j \in S_i
\end{cases}.\label{eqn:solution2}
\end{align}  
 
 The variable $u_{ij}^t=0$ represents that the constraint $x_{ij}^t$ is unenforced.
 Enforcing constraints never improves the objective value of an optimization problem and therefore $S_i\subseteq [n_i]$ is a maximal set where (\ref{eqn:solution2}) is feasible.  
 Moreover, it is straightforward to show that if $y_{ij}^t\geq y_{ik}^t$ then $x_{ij}^t\geq x_{ik}^t$. 
 Thus, greedily removing the lowest valued $y_{ij}^t$ from $\hat{S}_i=[n_i]$ until (\ref{eqn:solution2}) is feasible yields the optimal solution to (\ref{eqn:GD}). 
%%%%%%%%%%%%%%%%%%%%%%%%%%%%%%%%%%%%%%%%%%%%%%%%%%%%%%%%%%%%%%%%%%%%%%%%%
\section{Payoff Matrix Assumptions}\label{app:singularity}
The payoff matrix is in the form 
\begin{align}
A=\left[\begin{array}{c c}
a & b \\
c & d \\ 
\end{array}\right]
\end{align}
In this paper, we make three assumptions about $A$: $ad-bc=0$, $a>\max\{0,b,c\}$ and $d>\max\{0,b,c\}$. In order, we show that we may make these assumption without loss of generality.

In 2x2 games, if there is a unique fully mixed Nash equilibrium, then it is straight forward to show that player 2's equilibrium is 
\begin{align}
x^{NE}_2=\left(\frac{d-b}{a+d-b-c}, \frac{a-c}{a+d-b-c} \right)\label{eqn:NashCond}
\end{align}
and therefore $a+d-b-c\neq 0$, $d\neq b$ and $a\neq c$ when there is a unique fully mixed Nash equilibrium. Similarly, by analyzing player 1's Nash equilibrium, $d\neq c$ and $a\neq b$. Now consider the payoff matrix
\begin{align}
B=\left[\begin{array}{c c}
a+\frac{ad-bc}{a+d-b-c} & b+\frac{ad-bc}{a+d-b-c} \\
c+\frac{ad-bc}{a+d-b-c} & d+\frac{ad-bc}{a+d-b-c} \\ 
\end{array}\right]
\end{align}
The determinant of payoff matrix $B$ is zero.  Moreover, (\ref{eqn:FTRL}) is invariant to shifts in the payoff matrix, so for the purpose of the dynamics $\{x^t\}_{t=1}^\infty$, $A$ and $B$ are equivalent matrices. Thus, without loss of generality we may assume the payoff matrix is singular by shifting the matrix by a specific constant.

Next, we argue that we may assume $a>0$. Players 1 and 2 separately try to solve
\begin{align}\max_{x_1\in {\cal X}_1}\min_{x_2\in {\cal X}_2}x_1\cdot 
				\left[\begin{array}{c c}
				a & b \\
				c & d \\ 
				\end{array}\right]x_2
				&=\phantom{-}\max_{x_1\in {\cal X}_1}\min_{x_2\in {\cal X}_2}-x_1\cdot 
					\left[\begin{array}{c c}
					-a & -b \\
					-c & -d \\ 
					\end{array}\right]x_2\\
				&= -\max_{x_2\in {\cal X}_2}\min_{x_1\in {\cal X}_1}\phantom{-}x_2\cdot 
				\left[\begin{array}{c c}
				-a & -c \\
				-b & -d \\ 
				\end{array}\right]x_1.
\end{align}
Thus, by possibly switching the maximization and minimization roles between player 1 and player 2, we may assume $a>0$.

Next we show that we may assume $a>\max\{b,c\}$. If $a+d-b-c>0$ then (\ref{eqn:NashCond}) implies $a>c$ and, symmetrically, $a>b$ completing the claim.  
If instead, $a+d-b-c<0$, then through identical reasoning, $\min\{b,c\}>a>0$ and we can simply rewrite the payoff matrix as
\begin{align}\max_{x_1\in {\cal X}_1}\min_{x_2\in {\cal X}_2}x_1\cdot 
\left[\begin{array}{c c}
a & b \\
c & d \\ 
\end{array}\right]x_2
&=\max_{x_1\in {\cal X}_1}\min_{x_2\in {\cal X}_2}x_1\cdot 
\left[\begin{array}{c c}
b & a \\
d & c \\ 
\end{array}\right]x_2
\end{align}
With the new payoff matrix, $b+c-a-d>0$ implying $b>\max\{a,d\}\geq 0$ as desired.  Thus, we may assume $a>\max\{0,b,c\}$ by relabeling player 1's strategies. 

Finally, $ad-bc=0$ and $a>\max\{0,b,c\}$ implies $d>\max\{0,b,c\}$.  
The prior analysis argues  $a+d-b-c>0$.
Thus,  (\ref{eqn:NashCond}) implies $d>\max\{b,c\}$. 
Now for contradiction, suppose $d<0$.  This implies $0>d>\max\{b,c\}$ and $ad-bc<0$ a contradiction.  Therefore $d>\max\{0,b,c\}$. 
%%%%%%%%%%%%%%%%%%%%%%%%%%%%%%%%%%%%%%%%%%%%%%%%%%%%%%%%%%%%%%%%%%%%%%%%%
\section{Expressing the Convex Conjugate with the Transformed Payoffs}\label{app:dualapp}
We can express $x$ as
\begin{align}
x_{11}^t	&=\begin{cases}
				0 & \mbox{if } \eta\left(1-\frac{c-d}{a-b}\right)\frac{y_{11}^t}{2}+\frac{1}{2}\leq0\\
				1 & \mbox{if } \eta\left(1-\frac{c-d}{a-b}\right)\frac{y_{11}^t}{2}+\frac{1}{2}\geq1\\
				\eta\left(1-\frac{c-d}{a-b}\right)\frac{y_{11}^t}{2}+\frac{1}{2}			& \mbox{otherwise}
\end{cases}\\
x_{21}^t	&=\begin{cases}
				0 & \mbox{if } \eta\left(1-\frac{b-d}{a-c}\right)\frac{y_{21}^t}{2}+\frac{1}{2}\leq0\\
				1 & \mbox{if } \eta\left(1-\frac{b-d}{a-c}\right)\frac{y_{21}^t}{2}+\frac{1}{2}\geq1\\
				\eta\left(1-\frac{b-d}{a-c}\right)\frac{y_{21}^t}{2}+\frac{1}{2}			& \mbox{otherwise}.
			\end{cases}.
\end{align} 
Thus, (\ref{eqn:notconvex}) simplifies to
\begin{align}
	h^*_{1}(y_{1}^t)	&=	\begin{cases}
						y_{12}^t-\frac{1}{2\eta} & \mbox{if } x_{11}^t=0\\
						y_{11}^t-\frac{1}{2\eta} & \mbox{if } x_{11}^t=1\\
						\frac{\eta}{4}\left(y_{11}^t-y_{12}^t\right)^2+\frac{y_{11}^t+y_{12}^t}{\eta}-\frac{1}{4\eta}	\hspace{.32in}\phantom{hi}		& \mbox{otherwise}
					\end{cases}\\
					&=\begin{cases}
						\frac{c-d}{a-b}y_{11}^t-\frac{1}{2\eta} & \mbox{if } x_{11}^t=0\\
						y_{11}^t-\frac{1}{2\eta} & \mbox{if } x_{11}^t=1\\
						\frac{\eta}{4}\left(1-\frac{c-d}{a-b}\right)^2\left(y_{11}^t\right)^2+\frac{\left(1-\frac{c-d}{a-b}\right)y_{11}^t}{\eta}-\frac{1}{4\eta}			& \mbox{otherwise}
					\end{cases}
\end{align}
Symmetrically, 
\begin{align}
h^*_{2}(y_{2}^t)&= \begin{cases}
\frac{b-d}{a-c}y_{21}^t-\frac{1}{2\eta} & \mbox{if } x_{21}^t=0\\
y_{21}^t-\frac{1}{2\eta} & \mbox{if } x_{21}^t=1\\
\frac{\eta}{4}\left(1-\frac{b-d}{a-c}\right)^2\left(y_{21}^t\right)^2+\frac{\left(1-\frac{b-d}{a-c}\right)y_{21}^t}{\eta}-\frac{1}{4\eta}			& \mbox{otherwise}\\
\end{cases}
\end{align}
Unlike (\ref{eqn:notconvex}), we can easily verify $h^*$ is strongly convex when the strategy is fully mixed. In addition to allowing for a simpler analysis, 
\begin{align}
h^*_{1}(y_{1}^t)&=\begin{cases}
\frac{c-d}{a-b}y_{11}^t-\frac{1}{2\eta} & \mbox{if } x_{11}^t=0\\
y_{11}^t-\frac{1}{2\eta} & \mbox{if } x_{11}^t=1\\
\frac{\eta}{4}\left(1-\frac{c-d}{a-b}\right)^2\left(y_{11}^t\right)^2+\frac{\left(1-\frac{c-d}{a-b}\right)y_{11}^t}{\eta}-\frac{1}{4\eta}		& \mbox{otherwise}\\
\end{cases}\\
&=\begin{cases}
\alpha_{10}z_1^t-\beta_{10} & \mbox{if } z_{1}^t\leq0\\
\alpha_{11}z_1^t-\beta_{11} & \mbox{if } z_{1}^t\geq1\\
\gamma_1(z_1^t)^2 + \alpha_{1}z_1^t-\beta_{1}			& \mbox{otherwise}\\
\end{cases}\\
&= \bar{h}^*_1(z_1^t)
\end{align}
where $\alpha_{10}< 0, \alpha_{11}> 0,$ and $\gamma_1>0$.  Symmetrically,
\begin{align}
\bar{h}^*_2(z_{2}^t)
&=\begin{cases}
\alpha_{20}z_2^t-\beta_{20} & \mbox{if } z_{2}^t\leq0\\
\alpha_{21}z_2^t-\beta_{21} & \mbox{if } z_{2}^t\geq1\\
\gamma_2(z_2^t)^2 + \alpha_{2}z_2^t-\beta_{2}			& \mbox{otherwise}\\
\end{cases}
\end{align}
 with $\alpha_{20}< 0, \alpha_{21}> 0,$ and $\gamma_2>0$.
%%%%%%%%%%%%%%%%%%%%%%%%%%%%%%%%%%%%%%%%%%%%%%%%%%%%%%%%%%%%%%%%%%%%%%%%%
\section{Details of Theorem \ref{thm:regret}}
\subsection{Partitioning the Strategies and the Dual Space}\label{sec:partition}
	By assumption $a>\min\{0,b,c\}$ and $d>\min\{0,b,c\}$ (See Appendix \ref{app:singularity}).  
	This implies that both the strategies ($x_{11}$, $x_{21}$) and the transformed payoff vector $z$ will rotate clockwise about the Nash equilibrium in both continuous and discrete time as depicted in Figure \ref{fig:discrete}. 
	To formally show clockwise movement, assume $x^t_{11}\geq x^{NE}_{11},x^t_{21}\geq x^{NE}_{21}$ (upper right of the Nash equilibrium).  
	Then $x_1^t\cdot Ax_2^t\leq (1,0)\cdot Ax_2^t$ implying $x_{11}^t\leq x_{11}^{t+1}$.  
	Symmetrically, $x_{21}^t\geq x_{22}^{t+1}$ implying that if $x^t_{11}\geq x^{NE}_{11}$ and $x^t_{21}\geq x^{NE}_{21}$ then the strategies move clockwise or not at all. Similarly, clockwise movement can be shown for the other three cases.  
	A symmetric argument shows the transformed payoff vector $z$ also rotates clockwise.

To partition the strategies $\{x^t\}_{t=1}^T$, we begin by first partitioning the dual space $X^*$ into $4$ regions $Z_0,Z_1,Z_2,$ and $Z_3$. The visual representation of this partitioning is given in Figure \ref{fig:partition}.
\begin{align*}
\color{black}Z_0&\color{black}=\left\{z: z_1< 1,z_2\geq 1\right\}.\\
\color{black}Z_1&\color{black}=\left\{z: z_1\geq 1,z_2> 0\right\}.\\
\color{black}Z_2&\color{black}=\left\{z: z_1> 0,z_2\leq 0\right\}.\\
\color{black}Z_3&\color{black}=\left\{z: z_1\leq  0,z_2< 1\right\}.
\end{align*}

The partitioning $Z_0$, $Z_1$, $Z_2$, and $Z_3$, is not a proper partitioning.  As depicted in Figure \ref{fig:partition}, it lacks all payoff vectors that correspond to fully mixed strategies for both players.  However, by Theorem \ref{thm:boundary}, there exists a $B$ so that $x^t$ is not fully mixed for both players for all $t\geq B$.  Since $B$ is finite, the first $B$ strategies will shift the total regret by at most a constant and therefore can be disregarded in our analysis. 
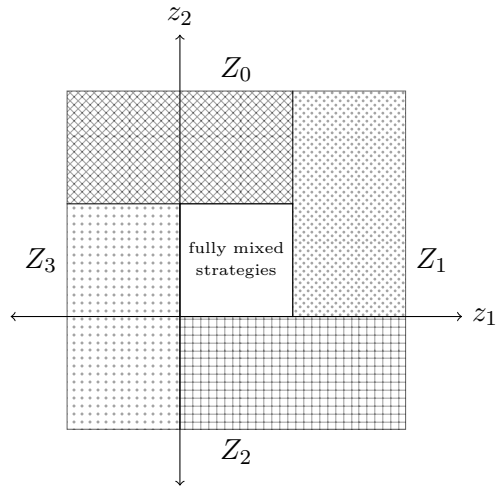
\begin{figure}[h]
\begin{center}
	\begin{tikzpicture}[scale=1.5]
	
	\def\shift{1}
	\def\axisshift{.5}
	\def\ETA{.5}
	
	\draw[<->] (0,0-\shift-\axisshift)--(0,1+\shift+\axisshift);
	\draw[<->] (0-\shift-\axisshift,0)--(1+\shift+\axisshift,0);
	
	\node[right] at (1+\shift+\axisshift,0) {$z_1$};
	\node[above] at (0,1+\shift+\axisshift) {$z_2$};
	
	\node at (.5,.6) {\tiny fully mixed};
	\node at (.5,.4) {\tiny strategies};
	
%	\draw[dashed] (0,-\shift)--(0,1+\shift);
%	\draw[dashed] (1,-\shift)--(1,1+\shift);
%	\draw[dashed] (-\shift,0)--(1+\shift,0);
%	\draw[dashed] (-\shift,1)--(1+\shift,1);	
	\draw (0,0) rectangle (1,1);

	\foreach \r in {}{
		\draw[domain=-.5/\ETA:.5/\ETA,smooth,variable=\x,blue] plot ({\ETA*\x+.5}, {.5+\ETA*(\r+3/(4*\ETA)-\ETA*\x*\x)});
		\draw[blue] (0,-\r/2)--(-\r/2,0);
		\draw[blue] (1,-\r/2)--(1+\r/2,0);
		\draw[blue] (0,1+\r/2)--(-\r/2,1);
		\draw[blue] (1,1+\r/2)--(1+\r/2,1);
		\draw[domain=-.5/\ETA:.5/\ETA,smooth,variable=\x,blue] plot ({\ETA*\x+.5}, 	{.5-\ETA*((\r+3/(4*\ETA)-\ETA*\x*\x)});
		\draw[domain=-.5/\ETA:.5/\ETA,smooth,variable=\x,blue] plot ({.5+\ETA*(\r+3/(4*\ETA)-\ETA*\x*\x)},{\ETA*\x+.5});
		\draw[domain=-.5/\ETA:.5/\ETA,smooth,variable=\x,blue] plot ({.5-\ETA*(\r+3/(4*\ETA)-\ETA*\x*\x)},{\ETA*\x+.5});
	}
	\draw[pattern=crosshatch, pattern color=black,opacity=.6] (-\shift,1)--(1,1)--(1,1+\shift)--(-\shift,1+\shift)--cycle;
	\node[black,above] at (.5,1+\shift) {$Z_0$};
	\draw[pattern=crosshatch dots, pattern color=black,opacity=.6] (1,0)--(1+\shift,0)--(1+\shift,1+\shift)--(1,1+\shift)--cycle;
	\node[black,right] at (1+\shift,.5) {$Z_1$};	
	\draw[pattern=grid, pattern color=black,opacity=.6] (0,0)--(1+\shift,0)--(1+\shift,-\shift)--(0,0-\shift)--cycle;
	\node[black,below] at (.5,-\shift) {$Z_2$};
	\draw[pattern=dots, pattern color=black,opacity=.6] (0,1)--(-\shift,1)--(-\shift,-\shift)--(0,-\shift)--cycle;
	\node[black,left] at (-\shift,.5) {$Z_3$};
	\end{tikzpicture}\caption{Visual Representation of $Z_0, Z_1, Z_2$ and $Z_3$.}\label{fig:partition}
\end{center}
\end{figure}

Since strategies move clockwise, in general the payoff vectors will move from region $Z_{i}$ to region $Z_{(i+1 \mod 4)}$.  If $\eta$ is large, then it is possible to move directly from $Z_i$ to $Z_{(i +2 \mod 4)}$.  While we consider such $\eta$ impractical, our analysis handles such cases and shows that after enough iterations, the payoff vectors never skip a region. Finally, we are able to define our partitioning over $\{x_t\}_{t=1}^T$. Let $B$ be as in the statement of Theorem \ref{thm:boundary} and let $Z(t)\in \{Z_0,Z_1,Z_2,Z_3\}$ be such that $z^t\in Z(t)$. 
\begin{align}
t_0&=\argmin_{t\geq B}\{z^t\in Z_0 \}\\
t_j&=\argmin_{t\geq t_{j-1}}\{z^{t}\notin Z({t-1})\} \ \forall j=1,2,...\label{eqn:dumb}
\end{align}
Finally, let $t_k=T+1$ where $k-1$ is the largest index that has a solution in (\ref{eqn:dumb}). 
The value $t_j$ represents the first time after $t_{j-1}$ that $z^t$ enters a new region. Our analysis now focuses on the time intervals created by these break points.  Specifically, we analyze $x^{t_j}, x^{t_j+1},...,x^{t_{j+1}-1}$ and $z^{t_j}, z^{t_j+1},...,z^{t_{j+1}-1}$

\subsection{Player Strategies Often Do Not Change}\label{sec:constant}
	In this section, we show that for each partitioning $\{t_j,...,t_{j+1}-1\}$ the strategies change at most a constant, $\kappa$, of times independent of the size of the partitioning, $t_{j+1}-t_j$. 
	This result is useful in two areas. 
	First, in the proof of Theorem \ref{thm:regret} it is used to show that $x^{t_j},...,x^{t_{j+1}-1}$ contributes to the regret by an amount proportional to $\kappa$. 
	Second, it is used in the proof of Lemma \ref{lem:LinearRadius} to show the total energy in the system increases by a constant in each partition;
	we show the energy only increases when the player strategies change and therefore, the energy increases at most $\kappa$ times in each partition. 

	\begin{lem}\label{lem:constant}
		There exists a $\kappa$ such that $|\{t\in \{t_j,..., t_{j+1}-1\}: x^{t}\neq x^{t+1}\}|\leq \kappa$ for all $j$. 
	\end{lem}
		\begin{proof}[Proof of Lemma \ref{lem:constant}]
		Without loss of generality, assume $z^{t_{j}},...,z^{t_{j+1}-1}\in Z_1$. 
		This implies $x_{11}^{t_{j}}=...=x_{11}^{t_{j+1}-1}=1$ and therefore
		\begin{align}
		y_{21}^{t+1}-y_{21}^t	&= [-a,-c]\cdot [x_{11}^t,1-x_{11}^t]\\	
		&= -a
		\end{align}
		for all $t=t_j, ..., t_{j+1}-1$.
		Thus, there must exist a constant $\delta_1>0$ such that $z_2^{t}-z_2^{t+1}=\delta_1$.
		
		By selection of $Z_1$, $z_2^t>0$ for all $t$.  
		Moreover, $x_{21}^t=1$ if $z_2^t\geq 1$. 
		Since $z_2^t$ decreases by $\delta_1$ in each iteration, $x_{21}^t\neq x_{21}^{t+1}$ iff $z_2^{t+1}<1$.  
		However, since $z_2^t-z_2^{t+1}=\delta_1$, there can only be at most $\kappa_1= \lceil 1/\delta_1 \rceil$ such $t$.  For regions $Z_0$, $Z_2$, and $Z_3$, there exist similar $\kappa_0,\kappa_2$, and $\kappa_3$. Taking $\kappa=\max\{\kappa_0,\kappa_1,\kappa_2,\kappa_3\}$ completes the proof of the lemma. 
	\end{proof}
\subsection{Energy Increases by $\Theta(1)$ in Each Partition}\label{sec:energy}
Next, we show the energy in the system increases by a constant each time $z^t$ moves into a new partition. 
Again, we use this result in two places. 
First, we use the result in the proof of Lemma \ref{lem:LinearSteps}, to show that $z^t$ moves from $Z_i$ directly to $Z_{({i+2}\mod 4)}$ at most a constant number of times. 
Second, we use the result in combination with Lemma \ref{lem:LinearSteps} to show $t_j\in \Theta(j^2)$ allowing us to conclude that $k\in \Theta(\sqrt{T})$ partitions are visited in $T$ iterations.

\begin{lem}\label{lem:LinearRadius}
	$r_{j+1}-r_j\in \Theta(1)$.
\end{lem}

The proof of Lemma \ref{lem:LinearRadius} relies on the observation that (\ref{eqn:GD}) is simply a 1st order approximation of (\ref{eqn:contGD}) as depicted in Figure \ref{fig:approx}. When neither $z_i^t\notin (0,1)$, the continuous time dynamics move in a straight line and therefore a 1st order approximation perfectly preserves the energy of the system. However, if $z_i^t\in (0,1)$ then by the strong convexity of $\bar{h}_i(z_i^t)$, the total energy of the system increases.  By Lemma \ref{lem:constant}, there are a constant number of $t$ where $z_i^t\in (0,1)$ for each partition and therefore the total energy increases by $O(1)$ in each partition. 
\begin{figure}[h]
	\begin{center}
		\begin{tikzpicture}[scale=1.25]
		\def\circlesize{1.5pt}
		\def\shift{1}
		\def\axisshift{.5}
		\draw[<->] (0,0-\shift-\axisshift)--(0,1+\shift+\axisshift);
		\draw[<->] (0-\shift-\axisshift,0)--(1+\shift+\axisshift,0);
		\node[left] at (0,1+\shift+\axisshift) {$z_2$};
		\node[below] at (1+\shift+\axisshift,0) {$z_1$};
		\draw[dashed] (0,-\shift)--(0,1+\shift);
		\draw[dashed] (1,-\shift)--(1,1+\shift);
		\draw[dashed] (-\shift,0)--(1+\shift,0);
		\draw[dashed] (-\shift,1)--(1+\shift,1);	
		\draw (0,0) rectangle (1,1);
		
		\draw[blue] (2,2)--(2.5,2);
		\node[right] at (2.5,2) {Continuous Time Dynamics};
		\fill[black] (2.25,1.7) circle[radius=\circlesize];
		\node[right] at (2.5,1.65) {Discrete Time Dynamics};
		
		\node[left] at (-1.5,2) {\phantom{Continuous Time Dynamics}};
		
		\newcommand\rad{1}
		\def\thresh{.5}
		\foreach \z in {.3}{
			\draw[domain=0:1,smooth,variable=\x,blue] plot ({\x}, {1+\rad-\x*\x+\x});
			\draw[blue] (0,-\rad)--(-\rad,0);
			\draw[blue] (1,-\rad)--(1+\rad,0);
			\draw[blue] (0,1+\rad)--(-\rad,1);
			\draw[blue] (1,1+\rad)--(1+\rad,1);
			\draw[domain=0:1,smooth,variable=\x,blue] plot ({\x}, 	{-\rad+\x*\x-\x});
			\draw[domain=0:1,smooth,variable=\x,blue] plot ({1+\rad-\x*\x+\x},{\x});
			\draw[domain=0:1,smooth,variable=\x,blue] plot ({-\rad+\x*\x-\x},{\x});
			
			\fill (\z, {1+\rad-\z*\z+\z}) circle[radius=\circlesize];
			\fill ({\z+\thresh}, {1+\rad-\z*\z+\z-(2*\z-1)*\thresh}) circle[radius=\circlesize];
			\draw (\z, {1+\rad-\z*\z+\z})--({\z+\thresh}, {1+\rad-\z*\z+\z-(2*\z-1)*\thresh});
		}
		\renewcommand\rad{1.25}
		\foreach \z in {.8}{
			\draw[domain=0:1,smooth,variable=\x,blue] plot ({\x}, {1+\rad-\x*\x+\x});
			\draw[blue] (0,-\rad)--(-\rad,0);
			\draw[blue] (1,-\rad)--(1+\rad,0);
			\draw[blue] (0,1+\rad)--(-\rad,1);
			\draw[blue] (1,1+\rad)--(1+\rad,1);
			\draw[domain=0:1,smooth,variable=\x,blue] plot ({\x}, 	{-\rad+\x*\x-\x});
			\draw[domain=0:1,smooth,variable=\x,blue] plot ({1+\rad-\x*\x+\x},{\x});
			\draw[domain=0:1,smooth,variable=\x,blue] plot ({-\rad+\x*\x-\x},{\x});
			
			\fill (\z, {1+\rad-\z*\z+\z}) circle[radius=\circlesize];
			\fill ({\z+\thresh}, {1+\rad-\z*\z+\z-(2*\z-1)*\thresh}) circle[radius=\circlesize];
			\draw (\z, {1+\rad-\z*\z+\z})--({\z+\thresh}, {1+\rad-\z*\z+\z-(2*\z-1)*\thresh});
		}
		\renewcommand\rad{.91}
		\foreach \z in {.3}{
			\draw[domain=0:1,smooth,variable=\x,blue] plot ({\x}, {1+\rad-\x*\x+\x});
			\draw[blue] (0,-\rad)--(-\rad,0);
			\draw[blue] (1,-\rad)--(1+\rad,0);
			\draw[blue] (0,1+\rad)--(-\rad,1);
			\draw[blue] (1,1+\rad)--(1+\rad,1);
			\draw[domain=0:1,smooth,variable=\x,blue] plot ({\x}, 	{-\rad+\x*\x-\x});
			\draw[domain=0:1,smooth,variable=\x,blue] plot ({1+\rad-\x*\x+\x},{\x});
			\draw[domain=0:1,smooth,variable=\x,blue] plot ({-\rad+\x*\x-\x},{\x});
			
			\fill (\z, {1+1-\z*\z+\z}) circle[radius=\circlesize];
			\fill ({\z-\thresh}, {1+1-\z*\z+\z-\thresh}) circle[radius=\circlesize];
			\fill ({\z-2*\thresh}, {1+1-\z*\z+\z-2*\thresh}) circle[radius=\circlesize];
			\draw (\z, {1+1-\z*\z+\z})--({\z-\thresh}, {1+1-\z*\z+\z-\thresh});
			\draw ({\z-2*\thresh}, {1+1-\z*\z+\z-2*\thresh})--({\z-\thresh}, {1+1-\z*\z+\z-\thresh});
		}
		\end{tikzpicture}\caption{Discrete Time is a 1st Order Approximation of Continuous Time.}\label{fig:approx}
	\end{center}
\end{figure}
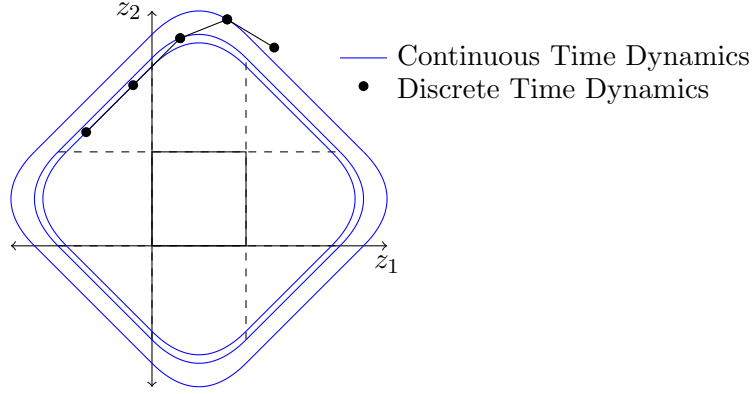

For the proof of Lemma \ref{lem:LinearRadius} is is useful to recall the following from Section \ref{sec:dual}:
\begin{align}
\bar{h}^*_1(z_1^t)&=\begin{cases}
\alpha_{10}z_1^t-\beta_{10} & \mbox{if } z_{1}^t\leq0\\
\alpha_{11}z_1^t-\beta_{11} & \mbox{if } z_{1}^t\geq1\\
\gamma_1(z_1^t)^2 + \alpha_{1}z_1^t-\beta_{1}			& \mbox{otherwise}\\
\end{cases}\\
\bar{h}^*_2(z_{2}^t)
&=\begin{cases}
\alpha_{20}z_2^t-\beta_{20} & \mbox{if } z_{2}^t\leq0\\
\alpha_{21}z_2^t-\beta_{21} & \mbox{if } z_{2}^t\geq1\\
\gamma_2(z_2^t)^2 + \alpha_{2}z_2^t-\beta_{2}			& \mbox{otherwise}\\
\end{cases}
\end{align}
where $\alpha_{i0}< 0, \alpha_{i1}> 0,$ and $\gamma_i>0$.

\begin{proof}[Proof of Lemma \ref{lem:LinearRadius}]
	Without loss of generality assume $z^{t_{j}},...,z^{t_{j+1}-1}\in Z_1$.
	Once again by selection of $Z_1$, $z_2^t>0$ and $x_{11}^t=1$ implying $z_1^t\geq 1$ for all $t=t_j,...,t_{j+1}-1$.  
	Let $R^t=\sum_{i=1}^2\bar{h}^*_i(z_i^t)$ be the total energy in the system in iteration $t$. 
	By \cite{GeorgiosSODA18}, the continuous time dynamics are captured by $\{z: \sum_{i=1}^2h^*_i(z_i)=R\}$ around the point $z^t$.
	When $z_1^t\geq 1$, the continuous time dynamics around $z^t$ are captured by
	\begin{align}
		R^t	&=\sum_{i=1}^2 \bar{h}_i^*(z_i)\\
			&=\bar{h}_2^*(z_2)+\alpha_{11}z_1-\beta_{11}
	\end{align}
	reducing to 
	\begin{align}
		z_1&= \frac{R^t+\beta_{11}-\bar{h}_2^*(z_2)}{\alpha_{11}}.
	\end{align}
	As observed earlier, (\ref{eqn:GD}) is simply a 1st order approximation of (\ref{eqn:contGD}) and therefore
	\begin{align}
		z_1^{t+1}	&=z_1^t-\frac{\nabla\bar{h}_2^*(z_2^t)}{\alpha_{11}}(z_2^{t+1}-z_2^t)\\
					&=z_1^t+\frac{\nabla\bar{h}_2^*(z_2^t)}{\alpha_{11}}\delta_1
	\end{align}
	where $\delta_1=z_2^t-z_2^{t+1}$ is shown to be constant in the proof of Lemma \ref{lem:constant}. We now examine the five possible locations for $z_2^t$ and $z_2^{t+1}$.  
	
	\textbf{Case 1: $z_2^t\geq1,z_2^{t+1}\geq1$}. We show there is no change to the energy in the system.  Since  $z_2^t\geq 1$, 
	\begin{align}
		z_1^{t+1}	&=z_1^t+\frac{\nabla\bar{h}_2^*(z_2^t)}{\alpha_{11}}\delta_1\\
					&=z_1^t+\frac{\alpha_{21}}{\alpha_{11}}\delta_1
	\end{align}
	
	The total energy in iteration $t+1$ is given by
	\begin{align}
		R^{t+1}	&= \sum_{i=1}^2 \bar{h}_i^*(z_i^{t+1})\\
				&=\alpha_{11}z_1^{t+1}-\beta_{11}+\alpha_{21}z_2^{t+1}-\beta_{21}\\
				&=\alpha_{11}\left(z_1^t+\frac{\alpha_{21}}{\alpha_{11}}\delta_1\right)-\beta_{11}+\alpha_{21}(z_2^{t}-\delta_1)-\beta_{21}\\
				&=\alpha_{11}z_1^t-\beta_{11}+\alpha_{21}z_2^{t}-\beta_{21}\\
				&= \sum_{i=1}^2 \bar{h}_i^*(z_i^{t})=R^t
	\end{align}
	and the energy in the system remains unchanged.

	\textbf{Case 2: $z_2^t\in (0,1),z_2^{t+1}\in (0,1)$}. 
	We show the energy increases by at least  ${\gamma_2}\delta_1^2$. 
	We begin with writing $z(\delta)$ as 
	\begin{align}
	z_1(\delta)	&=z_1^t+\frac{\nabla\bar{h}_2^*(z_2^t)}{\alpha_{11}}\delta\\
	&=z_1^t+\frac{2\gamma_2z_2^t+\alpha_2}{\alpha_{11}}\delta_1\\
	z_2(\delta) &=z_2^t-\delta
	\end{align}
	Therefore, $z^{t+1}=z(\delta_1)$. 
	Similarly, let $R(\delta)$ be energy associated with the point $z(\delta)$.  Formally, 
	\begin{align}
		R(\delta)	&= \sum_{i=1}^2 \bar{h}_i^*(z_i^{t+1}(\delta))\\
						&=\alpha_{11}z_1^{t+1}(\delta)-\beta_{11}+\gamma_2(z_2^{t+1}(\delta))^2+\alpha_2z_2(\delta)-\beta_2\\
						&=\alpha_{11}\left(z_1^t+\frac{2\gamma_2z_2^t+\alpha_2}{\alpha_{11}}\delta\right)-\beta_{11}+\gamma_2(z_2^{t}-\delta)^2+\alpha_2(z_2^{t}-\delta)-\beta_2
	\end{align}
	and $R(\delta_1)=R^{t+1}$ and $R(0)=R^t$. 
	Moreover $\frac{d^2R}{d\delta^2}=2\gamma_2>0$ and therefore $R(\delta)$ is strongly convex with parameter $2\gamma_2$. Thus,
	 \begin{align}
		R^{t+1}=R(\delta_1)	&\geq R(0)+R'(0)+{\gamma_2}\delta_1^2\\
							&= R^t+{\gamma_2}\delta_1^2
	\end{align}
	and the energy increases by at least ${\gamma_2}\delta_1^2$ completing Case 2. 
	
	\textbf{Case 3: $z_2^t\geq1,z_2^{t+1}\in (0,1)$}. The energy increases by at least ${\gamma_2}(1-z_{2}^{t+1})^2$.  This case follows identically to Case 2 by approximating $R(\delta_1)$ using strong convexity and $R(z_2^t-1)$.
	
	\textbf{Case 4: $z_2^t\in (0,1),z_2^{t+1}\leq 0$}. The energy increases by at least ${\gamma_2}(z_{2}^{t})^2$.  This case follows similarly to Cases 2 and 3. 
	
	\textbf{Case 5: $z_2^t\geq 1,z_2^{t+1}\leq 0$}. The energy increases by at least ${\gamma_2}$.  This case follows similarly to Cases 2-4. 
	
	We now can compute $r_{j+1}-r_j$. In each case, the increase in energy is bounded above since $z_2^t-z_2^{t+1}$ is bounded. Let $C_k$ be the number of times that Case $k$ occurs.  Case $1$ results in no change to the energy.  By Lemma \ref{lem:constant}, Case $2$ occurs at most $\kappa_1$ times.  Since $z_2^t$ is decreasing, Cases 3, 4, and 5 can occur at most once each.  Therefore $r_{j+1}-r_j\in \sum_{k=2}^5 C_k\cdot O(1)\leq (\kappa_1+3)\cdot O(1)\in O(1)$.  It remains to show $r_{j+1}-r_j\in \Omega(1)$.

	First suppose Case $2$ occurs at least once, then immediately we have $r_{j+1}-r_j\geq \gamma\delta_1^2 \in \Omega(1)$.
	If Case $2$ does not occur, then either Cases $3$ and $4$ occur, or Case $5$ occurs.  
	If Case $5$ occurs then $r_{j+1}-r_j\geq {\gamma_2}\in \Omega(1)$
	If Cases $3$ and $4$ occur but Case $2$ does not, only one $t$ is such that $z_2^t\in (0,1)$.  
	Thus, $r_{j+1}-r_j \geq \min_{z_2\in (0,1)}\{{\gamma_2}(z_{2})^2+{\gamma_2}(1-z_{2})^2\}=\frac{\gamma_2}{2}\in \Omega(1)$.  
	In all possibilities, $r_{j+1}-r_j\in \Omega(1)$ completing the proof of the lemma.
\end{proof}
\subsection{The Steps Per Partition are Proportional to the Energy}\label{sec:steps}

	In this section, we show that the number of steps in a partition is proportional to the total energy in the system.  
	We establish this by leveraging the connection between (\ref{eqn:contGD}) and (\ref{eqn:GD}). 
	Lemma \ref{lem:LinearSteps} is used in conjunction with Lemma \ref{lem:LinearRadius} to show a quadratic relationship between the total number of iterations and the number of partitions that the strategies have passed through. 
	This quadratic relationship directly leads to the $O(\sqrt{T})$ regret bound in Theorem \ref{thm:regret}. 

	\begin{lem}\label{lem:LinearSteps}
		$t_{j+1}-t_j\in \Theta(r_j)$. 
	\end{lem}

	\begin{proof}[Proof of Lemma \ref{lem:LinearSteps}]
		Without loss of generality, assume $z^{t_{j}},...,z^{t_{j+1}-1}\in Z_1$. 
		As in the proof of Lemma \ref{lem:constant}, there exists a constant $\delta_1>0$ such that $z_2^t-z_2^{t+1}=\delta_1$ for all $t=t_j, ..., t_{j+1}-1$.
		This implies $\delta_1(t_{j+1}-t_j)=z_2^{t_j}-z_2^{t_{j+1}}$.
		Thus, to prove  Lemma \ref{lem:LinearSteps} it suffices to show $z_2^{t_j}-z_2^{t_{j+1}}\in \Theta(r_j)$. 
		By definition of $t_{j+1}$, $z_2^{t_{j+1}-1}\in Z_1$ and therefore
		$0\geq z_2^{t_{j+1}} = z_2^{t_{j+1}-1}-\delta\geq -\delta$.
		Thus, $z_2^{t_j}-z_2^{t_{j+1}}\in \Theta(r_j)$ if and only if $z_2^{t_j}\in \Theta(r_j)$.
 
		To show $z_2^{t_j}\in \Theta(r_j)$ and complete the proof, we break the problem into 6 cases based on the location of $z^{t_j-1}$ as depicted in Figure \ref{fig:cases}. 
		The analyses for Cases 1-3 are similar and we show Cases 4-6 can never occur. 
		
		\begin{figure}[h]
			\begin{center}
				\begin{tikzpicture}[scale=2]
				
				\def\shift{1}
				\def\axisshift{.5}
				\def\ETA{.5}
				
				\draw[<->] (0,0-\shift-\axisshift)--(0,1+\shift+\axisshift);
				\draw[<->] (0-\shift-\axisshift,0)--(1+\shift+\axisshift,0);
				
				\node[right] at (1+\shift+\axisshift,0) {$z_1$};
				\node[above] at (0,1+\shift+\axisshift) {$z_2$};
				
				\node at (.5,.6) {\tiny fully mixed};
				\node at (.5,.4) {\tiny strategies};	
				\draw (0,0) rectangle (1,1);

				\foreach \r in {}{
					\draw[domain=-.5/\ETA:.5/\ETA,smooth,variable=\x,blue] plot ({\ETA*\x+.5}, {.5+\ETA*(\r+3/(4*\ETA)-\ETA*\x*\x)});
					\draw[blue] (0,-\r/2)--(-\r/2,0);
					\draw[blue] (1,-\r/2)--(1+\r/2,0);
					\draw[blue] (0,1+\r/2)--(-\r/2,1);
					\draw[blue] (1,1+\r/2)--(1+\r/2,1);
					\draw[domain=-.5/\ETA:.5/\ETA,smooth,variable=\x,blue] plot ({\ETA*\x+.5}, 	{.5-\ETA*((\r+3/(4*\ETA)-\ETA*\x*\x)});
					\draw[domain=-.5/\ETA:.5/\ETA,smooth,variable=\x,blue] plot ({.5+\ETA*(\r+3/(4*\ETA)-\ETA*\x*\x)},{\ETA*\x+.5});
					\draw[domain=-.5/\ETA:.5/\ETA,smooth,variable=\x,blue] plot ({.5-\ETA*(\r+3/(4*\ETA)-\ETA*\x*\x)},{\ETA*\x+.5});
				}
				\draw[pattern=crosshatch, pattern color=black,opacity=.6] (-\shift,1)--(1,1)--(1,1+\shift)--(-\shift,1+\shift)--cycle;
				\node[black,above] at (.5,1+\shift) {$Z_0$};
				\draw[pattern=crosshatch dots, pattern color=black,opacity=.6] (1,0)--(1+\shift,0)--(1+\shift,1+\shift)--(1,1+\shift)--cycle;
				\node[black,right] at (1+\shift,.5) {$Z_1$};	
				\draw[pattern=grid, pattern color=black,opacity=.6] (0,0)--(1+\shift,0)--(1+\shift,-\shift)--(0,0-\shift)--cycle;
				\node[black,below] at (.5,-\shift) {$Z_2$};
				\draw[pattern=dots, pattern color=black,opacity=.6] (0,1)--(-\shift,1)--(-\shift,-\shift)--(0,-\shift)--cycle;
				\node[black,left] at (-\shift,.5) {$Z_3$};
				\node[fill=white] at (.5,1.5) {Case 1};
				\node[fill=white] at (-.5,1.5) {Case 2};
				\node[fill=white] at (-.5,.5) {Case 3};
				\node[fill=white] at (-.5,-.5) {Case 4};
				\node[fill=white] at (.5,-.5) {Case 5};
				\node[fill=white] at (1.5,-.5) {Case 6};
				\draw(1,0)--(1,-1);
				\draw[dashed] (-\shift,1)--(-\shift-\axisshift,1);
				\node[left] at (-\shift-\axisshift,1) {$z_2=1$};
				\end{tikzpicture}\caption{Cases for Lemma \ref{lem:LinearSteps}.}\label{fig:cases}
			\end{center}
		\end{figure}
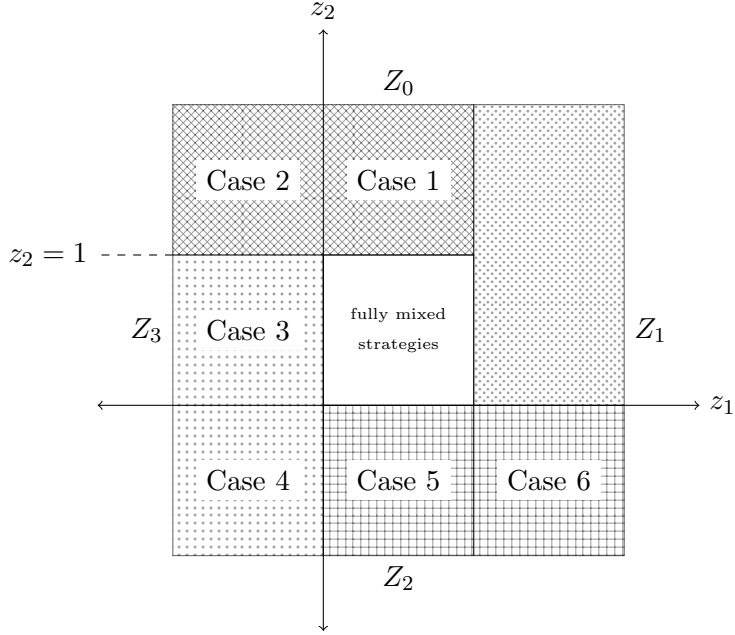

			\textbf{Case 1}: 
			Let $R^{t_j-1}$ be the energy at time $t_j-1$.
			Let $M_0=\{z: z\in Z_0, z_1\in [0,1], \sum_{i=1}^2\bar{h}_i^*(z_i)=R^{t_j-1}\}$.
			By definition, $z^{t_j-1}\in M_0$.
			Observe for $z\in M_0$, 
			\begin{align}
				R^{t_j-1}	&=\sum_{i=1}^2\bar{h}_i^*(z_i)\\
							&=\gamma_1(z_1)^2+\alpha_1z_1-\beta_1+\alpha_{21}z_2-\beta_{21}
			\end{align}
			and therefore 
			\begin{align}
				z_2&=\frac{R^{t_j-1}-\gamma_1(z_1)^2-\alpha_1z_1+\beta_1+\beta_{21}}{\alpha_{21}}\label{eqn:ref}
			\end{align}
			which is a concave function and therefore minimized at $z_1=0$ or $z_1=1$.
			Thus,
			\begin{align}
				z_2^{t_j-1} \geq \min_{z_1\in \{0,1\}} \frac{R^{t_j-1}-\gamma_1(z_1)^2-\alpha_1z_1+\beta_1+\beta_{21}}{\alpha_{21}}\in \Theta(R^{t_j-1}).
			\end{align} 
			Similar to the proof of Lemma \ref{lem:LinearRadius}, we compute $z_2^{t_j}$ from $z_2^{t_j-1}$: 
			\begin{align}
				z_2^{t_j}&=z_2^{t_j-1}-\frac{\nabla \bar{h}_1^*(z_1^{t_j-1})}{\alpha_{21}}(z_{1}^{t_j}-z_1^{t_j-1})\\
						&=z_2^{t_j-1}-\frac{ 2\gamma_1z_1^{t_j-1} -\alpha_1}{\alpha_{21}}\delta_0\\
						&\in z_2^{t_j-1}+\Theta(1)\in \Theta(R^{t_j-1})
			\end{align}
			since $z_1^{t_j-1}\in [0,1]$.  Finally, by Lemma \ref{lem:LinearRadius}, $r^{j-1}\leq R^{t_j-1} \leq r^{j}=r^{j-1}+\Theta(1)$ and therefore $z_2^{t_j}\in \Theta(r^j)$ completing Case 1.
			
			\textbf{Case 2:} This case follows identically to Case 1 using $\nabla \bar{h}_1^*(z_1^{t_j-1})=\alpha_{10}$.  
			
			\textbf{Case 3:} Similar to the proof of Case 1, $z_1^{t_j-1}\in -\Theta(R^{t_j-1})$ and $z_1^{t_j}=z_1^{t_j-1}+\Theta(1)$.  
			However, since $z^{t_j}\in Z_1$, $z^{t_j}_1\geq 1$ implying $R^{t_j-1}\in \Theta(1)$. 
			Therefore, by Lemma \ref{lem:LinearRadius}, $r^{j}\in \Theta(1)$.  
			Let $\delta_3>0$ be as in the proof of Lemma \ref{lem:LinearRadius}.  
			Since $z^{t_j-1}\in Z_3$ and $z_2^{t_j-1}\in [0,1]$, $z_2^{t_j}=z_2^{t_j-1}+\delta_3$ and $z_2^{t_j}\in \Theta(1)=\Theta(r^j)$ completing Case 3.  
			
			\textbf{Case 4, 5 and 6:} In Case 4, the sign of $\nabla \bar{h}_2^*(z_2^{t_j-1})$ implies $z_1^{t_j}<z_1^{t_j-1}<0$.  In Case 5, $z_1^{t_j}=z_1^{t_j-1}-\delta_2<1$ where $\delta_2>0$ is defined in the proof of Lemma \ref{lem:LinearRadius}.  In Case 6, the sign of $\nabla \bar{h}_1^*(z_1^{t_j-1})$ implies $z_2^{t_j}<z_2^{t_j-1}<0$.  All three cases contradict that $z_1^{t_j}\in Z_1$ completing Cases 4, 5, and 6.
			
			In all 6 cases, $z_2^{t_j}\in \Theta(r_j)$ implying $t_{j+1}-t_j\in \Theta(r_j)$ completing the proof.
	\end{proof}

%%%%%%%%%%%%%%%%%%%%%%%%%%%%%%%%%%%%%%%%%%%%%%%%%%%%%%%%%%%%%%%%%%%%%%%%%
\section{Convergence to the Boundary}\label{app:boundary}
\begin{proof}[Proof of Theorem \ref{thm:boundary}] The proof of convergence to the boundary follows similarly to the details for Theorem \ref{thm:regret}. By \cite{BaileyEC18}, there exists a constant $w>0$ and a $T$ such that $\min_{i\in\{1,2\}}\{|x_{i1}^t-x_{i1}^{NE}|\}\geq w$ for all $t\geq T$.  Similar to Theorem \ref{thm:regret}, we can then partition the dual space around the Nash equilibrium as follows: 
\begin{align*}
\color{black}Z_0&\color{black}=\left\{z: z_1< x_{11}^{NE}+w,z_2\geq x_{22}^{NE}+w\right\}.\\
\color{black}Z_1&\color{black}=\left\{z: z_1\geq x_{11}^{NE}+w,z_2> x_{22}^{NE}-w\right\}.\\
\color{black}Z_2&\color{black}=\left\{z: z_1> x_{11}^{NE}-w,z_2\leq x_{22}^{NE}-w\right\}.\\
\color{black}Z_3&\color{black}=\left\{z: z_1\leq x_{11}^{NE}-w,z_2< x_{22}^{NE}+w\right\}.
\end{align*}

\begin{figure}[h]
	\begin{center}
		\begin{tikzpicture}[scale=1.5]
		
		\def\shift{1}
		\def\axisshift{.5}
		\def\ETA{.5}
		\def\NEa{.45}
		\def\NEb{.5}
		\def\w{.37}
		\fill (\NEa,\NEb) circle[radius=1.5pt];
		\node[below ] at (\NEa,\NEb) {$x^{NE}$};
		\draw[<->] (0,0-\shift-\axisshift)--(0,1+\shift+\axisshift);
		\draw[<->] (0-\shift-\axisshift,0)--(1+\shift+\axisshift,0);
		
		\node[right] at (1+\shift+\axisshift,0) {$z_1$};
		\node[above] at (0,1+\shift+\axisshift) {$z_2$};

		\draw (\NEa, \NEb+.1)--(\NEa, 1.4+\shift);
		\draw (\NEa+\w, 1.1)--(\NEa+\w, 1.4+\shift);
		\node at (\NEa+\w/2,1.25+\shift) {$w$};
		\draw (\NEa,1.25+\shift)--(\NEa+.05,1.25+\shift);
		\draw (\NEa+\w,1.25+\shift)--(\NEa+\w-.05,1.25+\shift);
		
		%	\draw[dashed] (0,-\shift)--(0,1+\shift);
		%	\draw[dashed] (1,-\shift)--(1,1+\shift);
		%	\draw[dashed] (-\shift,0)--(1+\shift,0);
			\draw[dashed] (-\shift-\axisshift,1)--(1+\shift,1);
			\node[left]	at (-\shift-\axisshift,1) {$z_2=1$};
		\draw (0,0) rectangle (1,1);
		
		\draw[pattern=crosshatch, pattern color=black,opacity=.6] (-\shift,\NEb+\w)--(\NEa+\w,\NEb+\w)--(\NEa+\w,1+\shift)--(-\shift,1+\shift)--cycle;
		\node[black,above] at (-\shift/2+\NEa/2-\w/2,1+\shift) {$Z_0$};
		\draw[pattern=crosshatch dots, pattern color=black,opacity=.6] (\NEa+\w,\NEb-\w)--(1+\shift,\NEb-\w)--(1+\shift,1+\shift)--(\NEa+\w,1+\shift)--cycle;
		\node[black,right] at (1+\shift,1/2+\shift/2+\NEb/2-\w/2) {$Z_1$};	
		\draw[pattern=grid, pattern color=black,opacity=.6] (\NEa-\w,\NEb-\w)--(1+\shift,\NEb-\w)--(1+\shift,-\shift)--(\NEa-\w,0-\shift)--cycle;
		\node[black,below] at (.5+\shift/2+\NEa/2-\w/2,-\shift) {$Z_2$};
		\draw[pattern=dots, pattern color=black,opacity=.6] (\NEa-\w,\NEb+\w)--(-\shift,\NEb+\w)--(-\shift,-\shift)--(\NEa-\w,-\shift)--cycle;
		\node[black,left] at (-\shift,-\shift/2+\NEb/2-\w/2) {$Z_3$};
		\end{tikzpicture}\caption{Partitioning for Theorem \ref{thm:boundary}.}\label{fig:partition2}
	\end{center}
\end{figure}
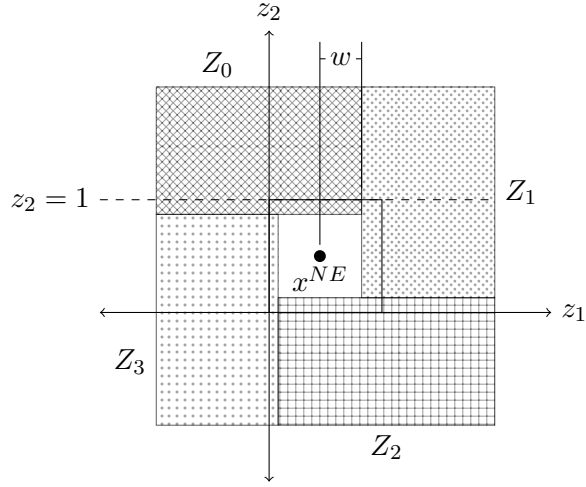

Once again the strategies rotate clockwise when updated with (\ref{eqn:GD}).  
Similar to Lemma \ref{lem:LinearRadius}, the energy increases by at least a constant in each iteration.  
By continuity of $\bar{h}_i^*$ and compactness, the energy $\sum_{i=1}^2 \bar{h}_i^*(z_i)$ is bounded above by $u$ when $z\in [0,1]^2$.  
Similar to Lemma \ref{lem:LinearSteps}, $z^t$ spends a bounded number of steps in a partition before moving onto the next partition. 
Since energy is increasing by a constant each time $z^t$ enters a new partition, there must exist an iteration $B$ when the energy exceeds $u$.  
Thus, for all $t\geq B$,  $z^t\notin [0,1]^2$ implying $x^t$ is on the boundary. \end{proof}
%%%%%%%%%%%%%%%%%%%%%%%%%%%%%%%%%%%%%%%%%%%%%%%%%%%%%%%%%%%%%%%%%%%%%%%%%
\section{Proof of Theorem \ref{lower_bound}}\label{sec:Lower}
In this section, we establish that the worst-case regret is exactly $\Theta(\sqrt{T})$. 
To establish this result, it remains to provide a game, learning rate, and initial condition $y^0$ where the regret is $\Omega(\sqrt{T})$.   
To establish this lower bound, we first express iteration $t$ uniquely with $t=\frac{n(n+1)}{2}+k$ for some $k\in \{0,...,n\}$.  
Using notation, we provide the exact position of the payoff vector, $y_i^t$ in each iteration. 
With this position, we compute the exact utility and regret through iteration $t$. 
Specifically, we show that in iteration $\frac{n(n+1)}{2}+k$, the total regret is $\frac{n}{2}+O(1)$.
To show these results, we use the game Matching Pennies with learning rate $\eta=1$ and initial payoff vectors $y_1^0=y_2^0=(1,0)$.

\begin{align}\tag{Matching Pennies Payoff Matrix}
\left( \begin{array}{r r} 1&-1\\-1&1\end{array}\right)
\end{align}

\begin{lem}\label{lem:ExactPosition}
	Consider the game Matching Pennies with learning rate $\eta=1$ and initial conditions $y_{1}^0=y_2^0=(1,0)$. In iteration  $t=\frac{n(n+1)}{2}+k$ where $k\in \{0,...,n\}$, player $i$'s payoff vector is given by  
	\begin{align*}
	y_{1}^{\frac{n(n+1)}{2}+k}=
	&\begin{cases}
		(1+k,-k) & \text{ if } n\equiv 0 \mod 4\\
		(1+n-k,-n+k) & \text{ if } n\equiv 1 \mod 4\\
		(-k,1+k) & \text{ if } n\equiv 2 \mod 4\\
		(-n+k,1+n-k) & \text{ if } n\equiv 3 \mod 4
	\end{cases}\\
	y_{2}^{\frac{n(n+1)}{2}+k}=
&	\begin{cases}
	(1+n-k,-n+k) & \text{ if } n\equiv 0 \mod 4\\
	(-k,1+k) & \text{ if } n\equiv 1 \mod 4\\
	(-n+k,1+n-k) & \text{ if } n\equiv 2 \mod 4\\
	(1+k,-k) & \text{ if } n\equiv 3 \mod 4
	\end{cases}.
	\end{align*}
\end{lem}
\begin{proof}
	The result trivially holds for the base case $t=n=k=0$.  
	We now proceed by induction and assume the results holds for $t=\frac{n(n+1)}{2}+k$ and show the result holds for $t+1$.
	We break the problem into four cases based on the remainder of $n/4$.  
	
	\textbf{Case 1: $n\equiv 0 \mod 4$}.  By the inductive hypothesis, $y_1^{t}= (1+k,-k)$ and $y_2^{t}=(1+n-k,-n+k)$. 
	Since $k\leq n$, $y_{11}^t\geq 1$ and $y_{21}^t\geq 1$.  
	Following similarly to Section \ref{sec:dual}, 
	\begin{align}
	x_{i1}^t= \begin{cases}
	1 & \text{if } y_{i1}^t\geq 1\\	
	0 & \text{if } y_{i1}^t\leq 0\\
	y_{i1}^t & \text{otherwise}
	\end{cases}.
	\end{align}
	Thus, $x_{1}^t=x_{2}^t=(1,0)$ implying 
	\begin{align}
	y_1^{t+1}&= y_1^t+Ax_2^t\\
	&=y_1^t+(1,-1)\\
	&=(1+k,-k)+(1,-1),\\
	y_2^{t+1}&= y_2^t-A^\intercal x_1^t\\
	&=y_2^t+(-1,1)\\
	&=(1+n-k,-n+k)+(-1,1).
	\end{align}
	If $k<n$, then $t+1=\frac{n(n+1)}{2}+[k+1]$ and $y_{1}^{t+1}=(1+[k+1],-[k+1])$ and $y_{2}^{t+1}=(1+n-[k+1], -n+[k+1])$ as predicted by the statement of the lemma. 
	If instead $k=n$, then $t+1= \frac{[n+1]([n+1]+1)}{2}$ where $[n+1]\equiv 1 \mod 4$.  
	Moreover,  $y_1^{t+1}=(1+k+1,-k-1)=([n+1]+1,-[n+1])$ and $y_2^{t+1}=(1+n-k-1,-n+k+1)=(0,1)$ again matching the statement of lemma.
	Thus, the inductive step holds for all values of $k$ when $n\equiv 0 \mod 4$.   
	
	\textbf{Case 2: $n\equiv 1 \mod 4$}. Since $k\in [0,n]$, $y_{11}^t\geq 1$ and $y_{21}^t \leq 0$. 
	Following identically to Case $1$, 
	$y_{1}^{t+1}=(1+n-k-1,-n+k+1)$ and $y_2^{t+1}=(-k-1,1+k+1)$ matching the statement of the lemma for all possible values of $k$.  
	
	\textbf{Case 3: $n\equiv 2 \mod 4$}. Since $k\in [0,n]$, $y_{11}^t\leq 0$ and $y_{21}^t \leq 0$. 
Following identically to the previous cases, 
$y_{1}^{t+1}=(-k-1,1+k+1)$ and $y_2^{t+1}=(-n+k+1,1+n-k-1)$ matching the statement of the lemma for all possible values of $k$.

	\textbf{Case 4: $n\equiv 3 \mod 4$}. Since $k\in [0,n]$, $y_{11}^t\leq 0$ and $y_{21}^t \geq 0$. 
Following identically to the previous cases, 
$y_{1}^{t+1}=(-n+k+1,1+n-k-1)$ and $y_2^{t+1}=(1+k+1,-k-1)$ matching the statement of the lemma for all possible values of $k$.  

In all four cases, the inductive hypothesis holds completing the proof of the lemma. 
\end{proof}

With the exact value of the payoff vector in each iteration, we can compute the cumulative utility.

{\color{black}
\begin{lem}\label{lem:ExactUtility}
	Consider the game Matching Pennies with learning rate $\eta=1$ and initial conditions $y_{1}^0=y_2^0=(1,0)$. In iteration  $t=\frac{n(n+1)}{2}+k$ where $k\in \{0,...,n\}$, player $1$'s cumulative utility is 
	\begin{align*}
	\sum_{s=0}^t x_1^s\cdot A x_2^s=
	\begin{cases}
	1-\frac{n}{2}+k & \text{ if } n\equiv 0 \mod 2\\
	\frac{n-1}{2}-k & \text{ if } n\equiv 1 \mod 2\\	\end{cases}.
	\end{align*}
\end{lem}
\begin{proof}
	We again proceed by induction.  
	The base case $t=n=k=0$ trivially holds. 
	We assume the result holds for $t =\frac{n(n+1)}{2}+k$ and show it holds for $t+1$.  
	Again, we break the problem into four cases based on the remainder of $n/4$. 
	
	\textbf{Case 1:} $n \equiv 0 \mod 4$. 
	First, we consider $k<n$. 
	Since $k<n$, $t+1$ is in the form $\frac{n(n+1)}{2}+[k+1]$ where $k+1\leq n$.
	Thus, by Lemma \ref{lem:ExactPosition}, $y_{11}^{t+1}=1+[k+1]\geq 1$ and $y_{21}^{t+1}=1+n-[k+1] \geq 1$ implying $x_{1}^{t+1}=x_2^{t+1}=(1,0)$. 
	therefore, 
	\begin{align}\sum_{s=0}^{t+1} x_1^s\cdot A x_2^s=  x_1^{t+1}\cdot A x_2^{t+1}+\sum_{s=0}^{t} x_1^s\cdot A x_2^s=1+1-\frac{n}{2}+k=1-\frac{n}{2}+[k+1].
	\end{align} This completes Case 1 when $k<n$.
	
	If instead $k=n$, then $t$ is in the form $\frac{[n+1]([n+1]+1)}{2}$ where $[n+1]\equiv 1 \mod 4$. 
	Similar to before, $y_{11}^{t+1}=1+n\geq 1$ and $y_{21}^{t+1}=0$ implying $x_1^{t+1}=(1,0)$ and $x_2^{t+1}=(0,1)$.  
	Therefore, 
	\begin{align}\sum_{s=0}^{t+1} x_1^s\cdot A x_2^s=  x_1^{t+1}\cdot A x_2^{t+1}+\sum_{s=0}^{t} x_1^s\cdot A x_2^s=-1+1-\frac{n}{2}+n=\frac{[n+1]-1}{2}.
	\end{align} 
	This completes Case 1 when $k=n$. 
	Thus the inductive hypothesis holds in Case 1.
	
	\textbf{Case 2:} $n\equiv 2 \mod 4$. This case holds similarly to Case 1.  The only difference is that $x_1^{t+1}=x_2^{t+1}=(0,1)$ when $k<n$ and $x_2^{t+1}=(1,0)$ when $k=n$ which does not change the value of $x_1^{t+1}\cdot Ax_2^{t+1}$. 
	
	\textbf{Case 3:} $n\equiv 1 \mod 4$. First consider $k<n$ implying $t+1$ is in the form $\frac{n(n-1)}{2}+[k+1]$ where $k+1\leq n$.  
	Similar to Case 1, $x_{1}^{t+1}=(1,0)$ and $x_{2}^{t+1}=(0,1)$.  
	This implies
	\begin{align}\sum_{s=0}^{t+1} x_1^s\cdot A x_2^s=  x_1^{t+1}\cdot A x_2^{t+1}+\sum_{s=0}^{t} x_1^s\cdot A x_2^s=-1+\frac{n-1}{2}-k=\frac{n-1}{2}-[k+1].
	\end{align}
	completing Case 3 when $k<n$. 
	
	If instead $k=n$, then $t$ is in the form $\frac{[n+1]([n+1]+1)}{2}$ where $[n+1]\equiv 2 \mod 4$.  
	This implies $x_1^{t+1}=(0,1)$ and $x_2^{t+1}=(0,1)$.  
	Therefore, 
	\begin{align}\sum_{s=0}^{t+1} x_1^s\cdot A x_2^s=  x_1^{t+1}\cdot A x_2^{t+1}+\sum_{s=0}^{t} x_1^s\cdot A x_2^s=1+\frac{n-1}{2}-n=1-\frac{[n+1]}{2},
	\end{align}
	matching the statement of the lemma.  
	This completes Case 3.
	
	\textbf{Case 4:} $n\equiv 3 \mod 4$. Case 4 follows from Case 3 in the same way that Case 2 follows from Case 1.  
	The hypothesis holds under all cases completing the proof of the lemma.   
\end{proof}
}
 
We now show that Matching Pennies with learning rate $\eta_1$ and initial conditions $y_{1}^0=y_2^0=(1,0)$ has regret $\Theta(\sqrt{T})$ when updated with (\ref{eqn:GD}). \vspace{.1in}

\begin{proof}[Proof of Theorem \ref{lower_bound}]
	Theorem \ref{thm:regret} establishes that the regret is $O(\sqrt{T})$.  
	To show that the regret is $\Omega(\sqrt{T})$, we show that in iteration $t=\frac{n(n+1)}{2}+k$, that player $1$'s regret is $\frac{n}{2}+O(1)$ completing the proof.  
	
	The total regret through iteration $t$ is given by
	\begin{align}
	\max_{x_1\in {\cal X}_1} x_1\cdot \sum_{t=0}^t Ax_2^s&=\max_{x_1\in {\cal X}_1} x_1\cdot (y_1^{t+1}-y^0)\\
	&= |y_{11}^{t+1}-1|
	\end{align}
	since $y_{11}^t-1=y_{12}^t$ for all $t$ by Lemma \ref{lem:ExactPosition}.
	
	If $k<n$, then $t+1=\frac{n(n-1)}{2}+[k+1]$ and   
	\begin{align}
	|y_{11}^{t+1}|- \sum_{s=0}^tx_1^s\cdot Ax_2^s=
	\begin{cases}
	[k+1]-(1-\frac{n}{2}+k)=\frac{n}{2} & \text{ if } n\equiv 0 \mod 4\\
	n-[k+1] - (\frac{n-1}{2}-k)=\frac{n-1}{2} & \text{ if } n\equiv 1 \mod 4\\
	[k+1]+1-(1-\frac{n}{2}+k)=\frac{n}{2}+1 & \text{ if } n\equiv 2 \mod 4\\
	n-[k+1]+1- (\frac{n-1}{2}-k)=\frac{n+1}{2} & \text{ if } n\equiv 3 \mod 4
	\end{cases}.
	\end{align}
	
	If $k=n$, then $t+1=\frac{[n+1]([n+1]-1)}{2}$ and   
	\begin{align}
	|y_{11}^{t+1}|- \sum_{s=0}^tx_1^s\cdot Ax_2^s=
	\begin{cases}
	[n+1]-(1-\frac{n}{2}+n)=\frac{n}{2} & \text{ if } n\equiv 0 \mod 4\\
	1 - (\frac{n-1}{2}-n)={\color{black}\frac{n+3}{2}} & \text{ if } n\equiv 1 \mod 4\\
	[n+1]+1-(1-\frac{n}{2}+n)=\frac{n}{2}+1 & \text{ if } n\equiv 2 \mod 4\\
	0- (\frac{n-1}{2}-n)=\frac{n+1}{2} & \text{ if } n\equiv 3 \mod 4
	\end{cases}.
	\end{align}
	In all cases, the total regret is $\frac{n}{2}+O(1)\in \Omega(\sqrt{t})$ completing the proof of the theorem.
\end{proof}

\end{document}